\documentclass[journal]{IEEEtran}
\usepackage[pdftex]{graphicx}
\graphicspath{{../pdf/}{../jpeg/}}
\DeclareGraphicsExtensions{.pdf,.jpeg,.png}
\usepackage[cmex10]{amsmath}
\usepackage{amsfonts}
\usepackage[english]{babel}
\usepackage{amssymb}
\usepackage{enumerate}
\allowdisplaybreaks
\usepackage{graphicx}
\usepackage[font={small}]{caption}
\usepackage{cite}
\usepackage{blindtext}
\usepackage{amsthm}
\usepackage{tcolorbox}
\usepackage{xcolor}
\usepackage{lipsum}
\usepackage{algorithm}
\usepackage{algorithmic}
\usepackage{array}
\usepackage{tabularx}
\usepackage{subcaption}

\newtheorem{lemma}{Lemma}
\newtheorem{Proposition}{Proposition}
\newtheorem{theorem}{Theorem}
\newtheorem{remark}{Remark}
\usepackage{color}

\begin{document}

\title{Energy-Efficient Hybrid Offloading for Backscatter-Assisted Wirelessly Powered MEC with   Reconfigurable Intelligent Surfaces}

	\author{S. Zargari, C. Tellambura, \textit{Fellow, IEEE}, and S. Herath,~\IEEEmembership{Member,~IEEE}
	\thanks{S. Zargari and C. Tellambura are with the Department of Electrical and Computer Engineering, University of Alberta, Edmonton, Canada, e-mails: \{zargari, ct4g\}@ualberta.ca.}
	\thanks{
	S. Herath is  with the Huawei Canada, 303 Terry Fox Drive, Suite 400, Ottawa, Ontario K2K 3J1 (e-mail:  sanjeewa.herath@huawei.com)}}
 	\maketitle

\begin{abstract}
 We investigate a wireless power transfer (WPT)-based backscatter-mobile edge computing (MEC) network with a reconfigurable intelligent surface (RIS).In this network,   wireless devices (WDs)  offload task bits and harvest energy,  and they can switch between   backscatter communication (BC) and active transmission (AT) modes. We exploit the  RIS to maximize energy efficiency (EE). To this end,  we optimize the time/power allocations, local computing frequencies, execution times, backscattering coefficients, and RIS phase shifts. This goal results in a multi-objective optimization problem (MOOP) with conflicting objectives. Thus,  we simultaneously maximize system throughput and minimize energy consumption via the Tchebycheff method, transforming into two single-objective optimization problems (SOOPs). For throughput maximization, we exploit alternating optimization (AO)  to yield two sub-problems. For the first one, we derive closed-form  resource allocations. For the second one, we design the RIS phase shifts via  semi-definite relaxation, a difference of convex functions programming, majorization minimization  techniques, and a penalty function for enforcing a rank-one solution. For energy minimization, we derive closed-form  resource allocations.  We demonstrate the gains over several benchmarks. For instance, with a  $20$-element RIS,  EE can be as high as 3 (Mbits/Joule),  a  150\%  improvement over the no-RIS case  (achieving only 2 (Mbits/Joule)).
\end{abstract}


\section{Introduction}\label{sec:introduction}

	\IEEEPARstart{T}{he} internet of things (IoT)  networks are forecasted to connect billions of wireless devices (WDs).  They are typically small-sized physical devices that are limited by throughput,  complexity, and power/processing capabilities.   Computation offloading to a mobile edge computing (MEC)  server and energy harvesting can thus alleviate these issues.  For example, WDs can harvest energy from a power beacon (PB) and execute their tasks bits locally and/or offload them to a MEC server \cite{Letaief,Lu,Bi1}. A real MEC testbed shows that computation offloading reduces latency up to 88\% and energy consumption up to 93\% \cite{Dolezal}. However, using active radio frequency (RF) signals for offloading task bits consumes a vast portion of the harvested energy. Thus, WDs backscatter (reflect) an ambient RF signal to transmit data, saving significant energy compared to active transmission (AT). Thus, backscatter communication (BC) will also be vital for sixth-generation (6G) wireless and beyond fifth-generation (5G) \cite{Saad}.

    However, simply integrating wireless power transfer (WPT), MEC, and BC is not enough. For instance, fading and other impairments (e.g., shadowing, path losses,  interference, and others) can impair the channel between a WD and the MEC server. Consequently, the poor reliability of bit offloading and energy harvesting (EH) efficiency can diminish the overall network performance. Thus, to alleviate these problems,   we suggest the use of the {reconfigurable intelligent surface (RIS)} technology, which improves energy efficiency (EE) \cite{Qingqing}. A  RIS panel comprises low-cost, adaptable elements that can smartly reflect incident signals by altering the phase shifts \cite{Wu}. {RISs} are nearly passive devices,  requiring only a small amount of energy that can be harvested from  RF signals, although this option is not pursued in this paper. As the RIS improves the wireless channel,  WDs can harvest more energy and offload task bits more reliably. Furthermore, RIS integrates well with  WPT to improve energy- and spectral- efficiency by intelligently adjusting the surrounding radio environment \cite{Yongjun_Xu}.

    Thus, there is a prima facie expectation of performance gains. \textbf{However, to the best of our knowledge, no prior work has studied the EE gains of RIS-BC-MEC networks, the focus of this paper.} Before the discussion of our contributions, we highlight some related works below.

	\subsection{Prior Works}
	In WPT-MEC networks, WDs harvest energy and use the harvest-then-transmit (HTT)  protocol to partially or fully offload their computation tasks  \cite{Bi1,Cui1,Chae}. Thus, each WD follows a binary (i.e., local or remote task execution) offloading policy \cite{Bi1}. This work maximizes all the WDs' weighted sum computation rate by jointly optimizing the mode selection (i.e., local computing or offloading) and the time allocation between WPT and task offloading. This yields a combinatorial optimization problem, and  \cite{Bi1} develops two effective solutions. Unlike \textit{binary} offloading, \textit{partial} offloading involves partitioning a task into two parts and offloading only one of them. This approach is considered in  \cite{Cui1} and \cite{Chae} to minimize the energy source and the total energy consumption of the MEC server by jointly optimizing the EH time allocation, energy transmit beamforming, partial computation offloading scheme, and central processing unit (CPU) frequencies.
	
	Offloading task bits to the MEC server is often done with active radios. But local oscillators and other RF components consume high powers and leave little energy available for the WDs' local computation, degrading the performance. To overcome this issue, each WD can also utilize BC instead. However, the rate and range of BC might be low. Thus, a combination of  BC and AT  can help each WD's offloading.  For instance,  for a multiple WD network with a  hybrid access point (HAP),  \cite{Zhong}  maximizes a reward function of MEC offloading by a workload allocation scheme and a price-based distributed time. Reference \cite{Cheng3} minimizes the HAP's energy consumption by optimizing the offloading and BC time allocations. Reference \cite{Lu} maximizes the weighted sum computation bits by jointly optimizing the transmit powers, time allocations, BC reflection coefficients, execution times, and local computing frequencies. 
	
	On the other hand, there are several investigations of RIS-MEC systems {\cite{Zhang,Nallanathan,Shanfeng,Zhiyang_Li,Xiaoyan_Hu,Yang_Liu}}. In \cite{Zhang}, a RIS helps two WDs offload to an access point (AP) connected with an edge cloud, where the sum delay is minimized by optimizing the RIS phase shifts and computation-offloading scheduling of the WDs. Reference \cite{Nallanathan} minimizes the latency by optimizing the edge computing resource allocation, multi-WD detection matrix, RIS phase shifts, and computation offloading volume. RIS-aided edge caching can minimize the network cost by optimizing the beamformers of the base station (BS), RIS phase shifts, and content placement at cache units \cite{Shanfeng}. In \cite{Zhiyang_Li}, the integration of BS and MEC with the aid of RIS is studied in which each WD can compute its task locally or partially/fully offload it to the BS. To support multiple WDs, the non-orthogonal multiple access (NOMA) protocol is considered. In particular,  the sum energy consumption of all WDs is minimized by joint optimization of the size of transmission data, power control,  transmission time/decoding order,  transmission rate, and RIS phase shifts. \cite{Xiaoyan_Hu} maximizes the total completed task-input bits of all WDs by jointly optimizing the receive beamforming vectors at the AP, energy partition strategies (local computing and offloading) of the WDs, and RIS phase shifts subject to limited energy budgets during a given time slot. In particular, three approaches are proposed to address this optimization problem, namely block coordinate descending (BCD) as a classical approach and two deep learning architectures based on the channel state information (CSI) and locations of the WDs by employing the BCD algorithm via supervised learning for training. Reference \cite{Yang_Liu} considers an AP equipped with a MEC to reduce the weighted sum of time and a RIS to improve the connections between WDs and AP. In particular, the earning of the MEC for loading computing is maximized by optimizing RIS phase shifts while guaranteeing a customized information rate for each WD. Furthermore, \cite{Tellambura} and \cite{Park} study the performance of the RIS in BC systems.  Optimizing carrier emitter's beamformers and RIS phase shifts to minimize the transmit power is investigated in \cite{Park}. 
	
	\subsection{Our Contributions}
	We briefly mention some previous contributions. In \cite{Ye}, a BC-aided WPT network is investigated; In \cite{Bi1,Cui1 ,Chae}, WPT MEC networks are studied; In \cite{Lu, Zhong, Cheng3 }, BC-aided WPT MEC networks are investigated. Furthermore, the integration of RIS with MEC or BC is studied in \cite{Zhang, Nallanathan, Shanfeng } and \cite{Tellambura, Park, Niyato}, respectively.  \textbf{Despite these advances, the performance of the BC-aided WPT MEC networks supported by RIS has not been addressed before to the best of our knowledge.}

	This paper fills this gap.  We investigate the deployment of a RIS to improve the performance of a BC-aided WPT MEC network (Fig. 1). The  PB, MEC server, and WDs are all single-antenna radios. Each WD harvests energy from PB's signals to recharge its battery and arbitrarily partitions the computation task into local computing and offloading parts. In particular, unlike \cite{Tellambura,Park,Ye,Niyato,Liang2019}, we consider hybrid WDs that operate in two modes,  namely BC and AT \cite{Lu,Zhong,Cheng3}. Furthermore, we assume that  WPT  and offloading tasks coincide over orthogonal frequency bands. In addition, a time division multiple access (TDMA) protocol enables the coordination of computation offloading, where different WDs offload their respective tasks to the MEC server over orthogonal time slots.    The main contributions and novelty aspects  of this paper are summarized as follows:

	\begin{enumerate}
		\item[$\bullet$]  Previous works  \cite{Bi1,Cui1,Chae,Zhong,Cheng3,Ye,Zargari,Niyato} consider the linear EH model. But it does not accurately capture the non-linear characteristics of practical EH circuits. To overcome this issue,  we consider a non-linear EH process  based on the rational model \cite{Chen}. 
		
		\item[$\bullet$] In \cite{Lu,Bi1} weighted sum computation rate of all the WDs is maximized whereas, in \cite{Cui1,Chae,Cheng3}, the total energy consumption is minimized. These formulations lead to different optimization problems. However, both the total throughput and energy consumption are the desirable objective functions, so we formulate an EE maximization problem. This problem requires jointly optimizing the time/power allocations, backscattering coefficients, local computing frequencies, execution times, and RIS phase shifts. Thus, our problem formulation is more general and subsumes these references.

		\item[$\bullet$] The formulated problem is a native multi-objective optimization problem (MOOP) that maximizes the total throughput while simultaneously minimizing energy consumption. Therefore, we resort to the Tchebycheff method to transform the MOOP into two single-objective optimization problems (SOOPs) instead of the widely used Dinklebach method \cite{Dinkelbach}. This method can yield a complete Pareto optimal set \cite{Mili,Mehrdad_Kolivand} compared to other approaches, such as $\epsilon$-method, weighted product method, and exponentially weighted criterion \cite{Miettinen}.
		
		\item[$\bullet$] We exploit alternating optimization (AO) to split the first SOOP (throughput maximization) into two sub-problems. For the first one, we derive closed-form resource allocations. For the second one, we design the RIS phase shifts via semi-definite relaxation (SDR), a difference of convex functions (DC) programming, and majorization minimization (MM) techniques. However, because of the interplay between RIS  and the BC WPT MEC network, RIS phase shifts appear quadratically, which is unlike existing works \cite{Zhong,Bi1,Cui1 ,Chae,Cheng3 ,Zhang ,Shanfeng, Tellambura,Park,Ye}. We reformulate these quadratic terms into linear forms, which is novel. An SDR solution is optimal if it is ranked one. But since it is not guaranteed, we augment the objective function to penalize the violation of the rank-one constraint. For the second SOOP (energy minimization), we only derive closed-form resource allocations since the energy consumption minimization problem is independent of the RIS phase shifts.

		\item[$\bullet$] Numerical results reveal the effectiveness of our proposed resource allocation. It yields notable gains compared to the other baseline schemes. More specifically, it can achieve  EE as high as 3 (Mbits/Joule), a  150\%  improvement over the no-RIS case supporting 2 (Mbits/Joule). Thus, RIS improves the system EE  of RIS-BC-MEC networks.  
		
	\end{enumerate}

	\textbf{Notations}: Boldfaced lowercase and uppercase letters represent column vectors and matrices,  respectively.  $(\mathbf{A})^T$,   $(\mathbf{A})^H$,   $\text{Tr}(\mathbf{A})$,  and $\text{Rank}(\mathbf{A})$ denote denote the transpose, Hermitian conjugate transpose,  trace, and rank of matrix $\mathbf{A}$, respectively. $\mathbf{A}\succeq\mathbf{0}$ indicates a positive semidefinite matrix. A vector is diagonalized by the operator $\text{diag}(\cdot)$.  A  circularly symmetric complex Gaussian (CSCG) random vector with mean $\boldsymbol{\mu}$ and covariance matrix $\boldsymbol{C}$ is denoted by $\sim \mathcal{C}\mathcal{N}(\boldsymbol{\mu},\,\boldsymbol{C})$. $\mathbb{E}(\cdot)$ is the statistical expectation. $\mathbb{C}^{M\times N}$ represents an $M\times N$ dimensional complex matrices, $\mathbb{R}^{N\times 1}$ denotes an $N\times 1$ dimensional real vectors $\|\mathbf{a}\|$ and $|b|$ denote the Euclidean norm of $ \mathbf{a} $  and the magnitude of  complex number $b$, respectively. $\nabla_\mathbf{x}f(\mathbf{x})$ denotes the gradient vector of function $f(\mathbf{x})$ with respect to $\mathbf{x}$. Finally, $[x]^+ = \max\{x,0\}$.

	\section{System Model}
	\begin{figure}[t]
		\centering
		\includegraphics[width=3.5in]{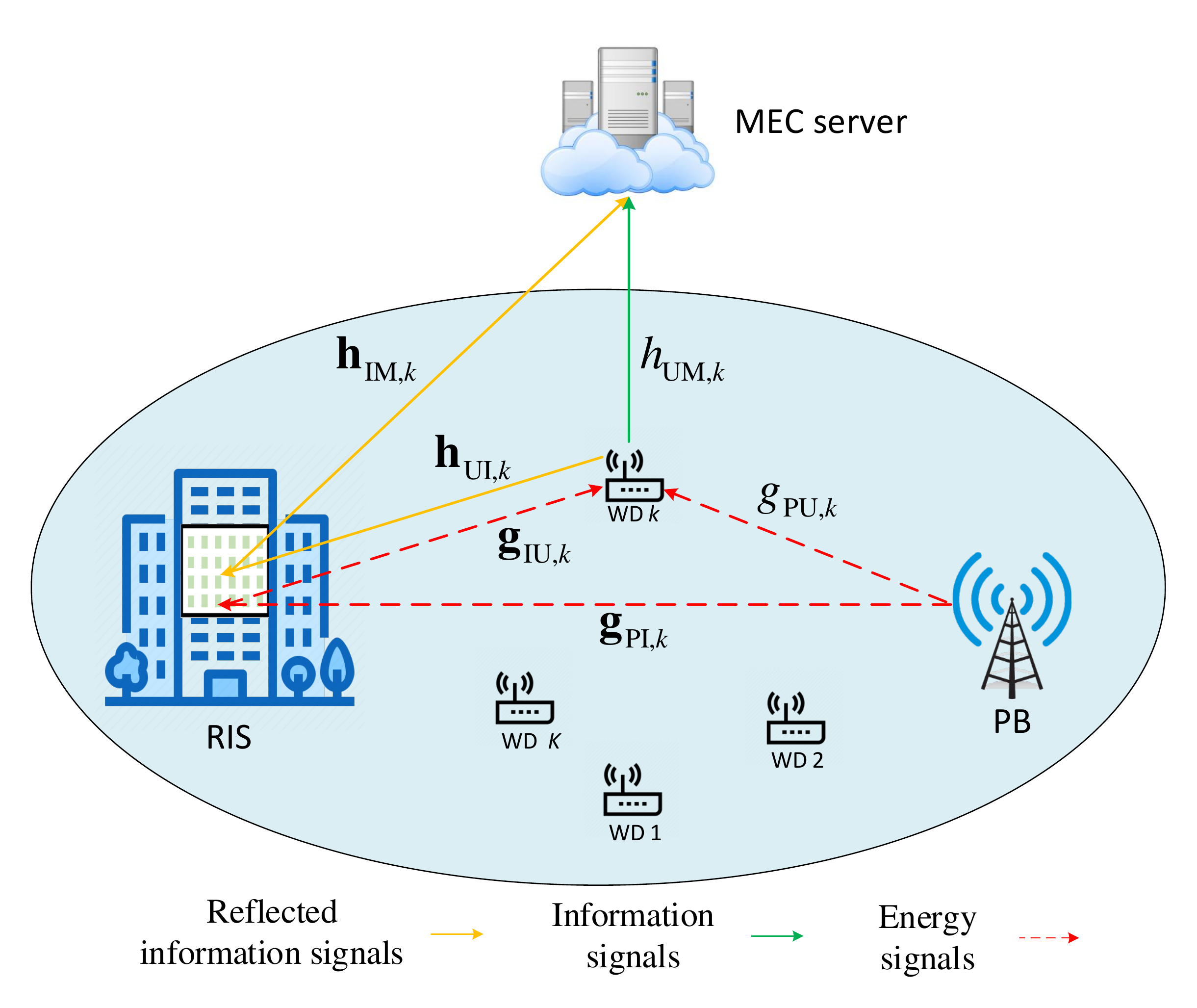}
		\caption{\small {Hybrid offloading  RIS-BC-MEC  network.}}
	\end{figure} 
	\subsection{System Setup}
	The system model (Fig. 1) comprises one PB, one MEC server, a RIS, and multiple WDs indexed by $\mathcal{K}=\{1,..., K\}$. Each WD  harvests energy to recharge its battery and has a  BC circuit, AT   circuit, and  EH module\footnote{Each WD has separate computing and offloading circuits and thus can perform local computation and task offloading simultaneously \cite{Bi1}.}. The PB transmits the energy signal. Each WD  harvests from it and uses the HTT protocol to partially offload task bits to the MEC server.  
    The RIS helps this process. Indeed, EH and the RIS can work in tandem to help WDs achieve data rates and lengthen the battery life. The PB, MEC server, and all WDs are equipped with a single antenna, while the RIS has $N$ reflecting elements. Hereafter, we denote the PB, each WD, the RIS, and the MEC server by $P$, $U$, $I$, and $M$, respectively.

	We consider the partial offloading policy by assuming that the task bits of each task are bit-wise independent \cite{Lu}. Also, we assume that each WD can adjust its computing frequency using the dynamic voltage scaling (DVS) technology\footnote{Note that performing EH compensates for the increasing hardware cost of the realization of DVS in each WD's circuit.} \cite{Bi1}. 
		\begin{figure}[t]
		\centering
		\includegraphics[width=3.5in]{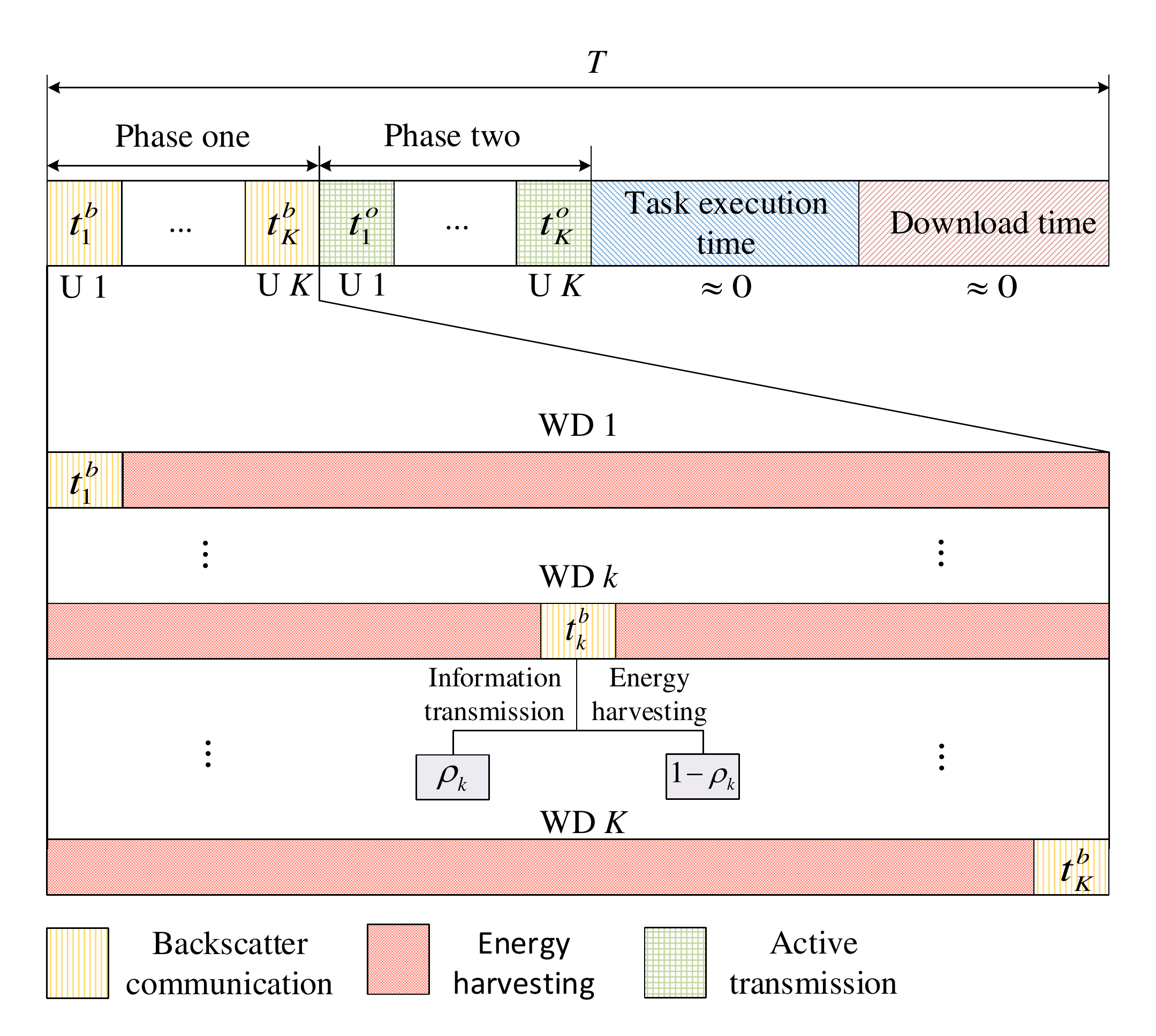}
		\caption{\small  {The time frame and the  different phases.}}\label{fig22}
	\end{figure}
	In the following, we describe the time frame of the system completely under four phases: 
	\begin{enumerate}
	 \item The PB transmits the energy signal.  Each WD harvests energy and offloads (via BC)  task bits to the MEC server simultaneously within its time slot and also harvests energy when other WDs start to offload. Indeed, each WD modulates its data over the received signal by tuning its load impedance into a number of states to map its data bit on the incident RF signal. For instance, in binary phase-shift keying (BPSK), the tag switches between two load impedances to generate bit “0” or “1” by matching or mismatching with antenna impedance, indicating absorbing or reflecting, respectively \cite{Fatemeh_Rezaei}.
	 \item The PB stops transmitting. The  WDs switch to the AT  mode and offload their task bits. 
	 \item The MEC server executes all the received computation tasks.
	\item The MEC server downloads the computation results to the WDs\footnote{The MEC server is high performance one, and the computation results are short. Thus,  downloading/computation times are negligible.}.
		\end{enumerate}

	Fig. 2 illustrates the time frame ($T$ duration)  for different events.  The $k$-th WD, $k\in\mathcal{K}$,  has the following parameters. 
	\begin{enumerate}
	    \item It has  BC time allocation of $t^b_k$ in 
	 phase one.
	 \item  It has  AT  time allocation and  transmit power $t^o_k$ and $p_k$, respectively, in the second phase.
	 \item It uses execution time, $0 \leq \tau_k \leq T$, and  local computing frequency, $f_k$, for performing local computation.
	\end{enumerate}   
	The transmit power of PB in the first phase is given by $P_0$. Finally, $g_{\text{PU},k}\in\mathbb{C}$,~$\mathbf{g}_{\text{PI}}\in\mathbb{C}^{N\times 1}$,~and $\mathbf{g}_{\text{IU},k}\in\mathbb{C}^{N\times 1}$ are the channels from the PB to WD $k$,~the PB to the RIS,~and the RIS to WD $k$,~respectively.~$h_{\text{UM},k}\in\mathbb{C}$,~$\mathbf{h}_{\text{UI},k}\in\mathbb{C}^{N\times 1}$,~and $\mathbf{h}_{\text{IM}}\in\mathbb{C}^{N\times 1}$ are the channels from WD $k$ to the MEC,~WD $k$ to the RIS,~and the RIS to the MEC,~respectively.

It should also be noted that BC is only suitable for short-range communications due to round-trip path loss. For this reason,  a RIS platform can improve spectral efficiency and leverage the phase shifts array gain to enhance the communication range.  Furthermore, to reduce the power consumption of WDs spent on data processing, we consider a MEC network. Accordingly, the applications of this system setup include smart homes/cities, healthcare, wearables, radio frequency identification (RFID),  manufacturing,  retail, warehouses, and others \cite{Jameel2019}  where the high computational task bit is required and/or round-trip path attenuation is severe.

\begin{remark}
The basic BC setup comprises a WD (tag) and a reader (receiver). The tag communicates with the reader via reflecting and modulating its data over the  RF carrier. Based on the RF source,  BC configurations can be monostatic, bistatic, or ambient. In the monostatic BC setup, the reader provides the unmodulated RF carrier to activate the WD. However, this setup results in low coverage due to the dyadic path loss \cite{Fabio_Bernardini}. In the bistatic configuration, the  RF emitter (or PB) is separated from the reader, increasing the communication range \cite{Abdelmoula_Bekkali}. Lastly, ambient BC uses  RF signals from ambient unpredictable/uncontrollable RF sources such as WiFi or TV signals \cite{Xiao_Lu}.  Out of these, we consider the bistatic BC configuration.  Accordingly, not only does PB act as an RF emitter, but it also helps the WDs in providing sufficient energy based on the HTT protocol to perform AT.	\end{remark}
\begin{remark}
 WDs in the BC system can be categorized into passive, semi-passive, and active tags depending on how they obtain power and utilize it. A passive tag only relies on the received power from the reader (RF emitter or PB) to operate, while a semi-passive tag has a  battery, albeit not without the cost of the extra weight and size. However, an active tag has  RF components and a battery similar to an active transmitter. As most applications demand passive and semi-passive tags because of their low cost,  low power consumption,  and non-generation of radio noise \cite{Jameel2019}, we focus on such semi-passive tags. 
\end{remark}

	\subsection{Channel Model}
		This system model presupposes perfect knowledge of all CSI\footnote{The RIS system may acquire CSI in two ways, depending on whether the RIS elements are equipped with receive RF chains or not. If the RIS has RF chains, traditional techniques can be used to estimate the PB-RIS, RIS-WDS, and RIS-MEC channels. If the RIS does not, uplink pilots and  RIS reflection patterns can be designed to estimate the channel links \cite{Qingqing, Wu}. Also, the algorithms designed in this letter can be extended to a robust optimization algorithm to reduce the effect of channel errors, e.g., \cite{Shahrokh_Farahmand}. } by using two smart RIS controllers between PB-RIS and RIS-MEC server, and all the channels experience quasi-static fading. These assumptions are standard throughout the literature (e.g., see \cite{Zargari1, Park, Niyato, Nallanathan}) and fairly realistic. With these two assumptions, our results identify an upper bound on the system performance. However, works such as \cite{Mishra,Lin} develop the estimation of channels involving the PB, WD, and the MEC. In contrast, one can estimate WD-RIS reflection channels by switching off the reflectors one at a time.  Note that the PB serves only as an RF power source, so the RF signal sent by the PB can be known to the MEC server. By utilizing the channel estimation approaches introduced in \cite{Ye}, the MEC server can find all the instantaneous CSI. In this way, the MEC can eliminate the interference caused by the PB via performing the successive interference cancellation (SIC). Based on the known CSI and performing SCI, the optimal resource allocation policy can be determined, which is then forwarded to the WDs and the PB which is reasonable since the energy signal transmitted by the PB can be predefined and known by the MEC\footnote{Indeed, MEC server first decodes the transmitted signals of the PB and then removes them from the received signal to get its information successively. Thus, the optimal resource allocation policy is then forwarded to the WDs and the PB.} \cite{Ye,Lu}.

	Furthermore, the RIS consists of $N$ reflecting elements, denoted by $\theta_n\in[0,2\pi)$, $\forall n\in\mathcal{N}= \{1,...,N\}$. For maximum reflection efficiency, the amplitude gain of each RIS element is set to one, i.e., $|\theta_n|^2=1$ $\forall n$. Let us denote the vector including the phase shifts of all reflecting elements as ${\boldsymbol{\theta}} = [\theta_1,\theta_2,...,\theta_N]^T$. Accordingly, we define the RIS phase shift matrix as $\boldsymbol{\Theta}= \text{diag} (e^{j\theta_1},e^{j\theta_2},...,e^{j\theta_N})$. Moreover, the cascade channel links between the PB-to-RIS-to-WD $k$ and the $k$-th WD-to-RIS-to-MEC are given by $\mathbf{g}_{\text{PI}}^H\boldsymbol{\Theta}\mathbf{g}_{\text{IU},k}$ and $\mathbf{h}_{\text{UI},k}^H\boldsymbol{\Theta}\mathbf{h}_{\text{IM}}$, respectively. Thus, the overall cascaded channel links from the PB-to-WD $k$ and the $k$-th WD-to-MEC are given by $g_k\triangleq g_{\text{PU},k}+\mathbf{g}_{\text{PI}}^H\boldsymbol{\Theta}\mathbf{g}_{\text{IU},k}$ and $h_k\triangleq h_{\text{UM},k}+\mathbf{h}_{\text{UI},k}^H\boldsymbol{\Theta}\mathbf{h}_{\text{IM}}$, respectively.
	
	\subsection{Signal Model}
	\subsubsection{First phase}
 In this phase, each WD splits the received signal of the PB into two parts by applying the backscattering coefficient $\rho_k$, which takes a value between “0” and “1”. The first part is utilized for  EH, and the second part is used by the WD to backscatter task bits to the MEC server. 
 
 Due to this arrangement, the received signal at the MEC server for executing the computation tasks is given by $
	y_{k,1}(t)= \sqrt{\rho_kP_0}h_kg_kx_1(t)+z_1(t),~\forall k,$ where  the $k$-th WD's  information symbol is  $x_1(t)\sim \mathcal{C}\mathcal{N}(0,\,1)$,
	and $z_1(t)$ is the received circularly symmetric complex Gaussian (CSCG) noise with zero mean and variance $\sigma^2$. The achievable offloading throughput of the BC mode  in the first phase can thus  be expressed as
	\begin{equation}\label{1}
	\Gamma^b_k=t^b_k\log_2(1+\frac{\zeta\rho_kP_0|h_k|^2|g_k|^2}{\sigma^2}),~\forall k,
	\end{equation}
	where $\zeta$ measures the signal-to-noise ratio (SNR) gap between practical modulation and coding schemes and depends on acceptable information outage and bit error rate (BER) \cite{Kim2017}. This also can be interpreted as a penalty for not utilizing ideal capacity-achieving Gaussian codes \cite{Liang2019}.  
	
	 With fraction  $1-\rho_k$ reserved for  EH, the harvested energy\footnote{We do not consider the harvested energy from the received noise, which is negligible compared to the received signal power.} at WD $k$ according to the linear model  is given by $E^h_k = (1-\rho_k)\eta_k t^b_kP_0g_k,~\forall k, $
	where $\eta_k$ is the energy conversion efficiency \cite{Zargari}. However, the linear model ignores the non-linear behavior of actual EH circuits, which motivates the development of practical non-linear EH models \cite{Boshkovska}. For instance, based on the sigmoidal function, a parametric non-linear EH model is given in  \cite{Boshkovska}. Its parameters are associated with the circuit properties and can be obtained by a curve fitting tool \cite{Boshkovska}. Nevertheless, this non-linear model makes the optimization problem of interest mathematically intractable.
	
	Thus, we use a more simplistic, yet non-linear, model,  namely  the rational model \cite{Chen}. This model yields the harvested power as  $f(x)=\frac{a_kx+b_k}{x+c_k}-\frac{b_k}{c_k}$, where parameters $a_k$, $b_k$, $c_k$ are determined by standard curve fitting with measured data. Accordingly, the amount of harvested energy during time $t^b_k$ in the first phase (Fig. \ref{fig22}) for WD $k$, $ k \in \mathcal K, $ can be written as:
	\begin{equation}
	E^b_k =\left(\frac{a_k(1-\rho_k)P_0|g_k|^2+b_k}{(1-\rho_k)P_0|g_k|^2+c_k}-\frac{b_k}{c_k}\right)t^b_k.
	\end{equation}
	Each WD also continues to harvest energy while other WDs take turns to offload task bits to the MEC server via BC in the first phase (Fig. \ref{fig22}). This amount of energy for WD $k$ can be stated as $\sum_{\scriptstyle i= 1, i \ne k}^KP^b_kt^b_i$. Finally, at the end of the first phase, the total harvested energy at WD $k$, $ k \in \mathcal K,$ can be written as $E^t_k =E^b_k+\sum_{\scriptstyle i= 1, i \ne k}^KP^b_kt^b_i$, where $P^b_k=\frac{a_kP_0|g_k|^2+b_k}{P_0|g_k|^2+c_k}-\frac{b_k}{c_k}.$
	\subsubsection{Second phase}
	The  received signal at the MEC server can now be written as 
	\begin{equation}
	    y_{k,2}(t)=\sqrt{p_k}h_kx_2(t)+z_2(t),~\forall k,
	\end{equation} where the $k$-th WD transmits   information symbol $x_2(t)\sim \mathcal{C}\mathcal{N}(0,\,1)$.
	The other term,  $z_2(t)$,  is the CSCG noise with zero mean and variance $\sigma^2$. Without loss of generality, we assume that the  noise variances in the first and second phases are the same. Consequently, the offloading throughput of the $k$-th  WD, $ k \in \mathcal K,$ via AT can be represented as
	\begin{equation}
	\Gamma^o_k=t^o_k\log_2(1+\frac{p_k|h_k|^2}{\sigma^2}). 
	\end{equation}
	Accordingly, the total offloading throughput of the $k$-th WD,   $ k \in \mathcal K,$ is given by $\Gamma^t_k = \Gamma^b_k + \Gamma^o_k.$ 
	
	Next,  the energy consumption  of the WD and the number of  local bit computations of  the WD   can be modeled  as ${E}_{\text{Cal},k} = \epsilon_kf_k^3\tau_k,~\forall k$, and ${\Gamma}_k =\frac{\tau_kf_k}{C_\text{cpu}},~\forall k$, respectively,
	where $C_\text{cpu}$ is the number of CPU  cycles required for calculating one bit,  and $\epsilon_k$ is the effective capacitance coefficient of the $k$-th WD's CPU \cite{Cui1}.  Finally, the total system throughput,  total offloading bit rate  plus the total computation bit rate, can be expressed as $R_{\text{sum}}=\sum_{\scriptstyle k= 1}^K(\Gamma^t_k+ {\Gamma}_k).$  
	
	\subsection{Energy Consumption Model}
	\subsubsection{First phase}  During this phase, the energy  consumption of each WD follows a linear model since it only offloads task bits to the MEC server without the need for complex circuit operations. Thus, the energy  consumption of the $k$-th WD, $k\in\mathcal{K}$, at the  constant circuit energy rate is  given by
	\begin{equation}
	  E_{1,k}=P_{c,k}t^b_k,\forall k, 
	\end{equation}
	where $P_{c,k}$ describes the circuit energy rate of the $k$-th WD, $k\in\mathcal{K}$.  
	
	\subsubsection{Second phase}
	In this phase, the energy consumed by the $k$-th WD, $ k\in\mathcal{K}$,  is more complicated since it operates in the  AT mode and performs local computing. Thus, its   energy consumption  can be modeled as 
	\begin{equation}
	   E_{2,k} = \frac{p_k}{\delta}t^o_k + p_{c,k}t^o_k+\epsilon_kf_k^3\tau_k,~\forall k,
	\end{equation}
	where $p_{c,k}$ denotes the circuit energy rate of the $k$-th WD in the AT mode, and $\delta$ is the power amplifier efficiency. In practice, the energy consumption  must  meet ${E}_{1,k} +E_{2,k} \leq E^t_k +Q_k,~\forall k$,  an  energy causality constraint.
	Furthermore, it is assumed that each WD has a certain amount of initial energy $Q_k$ \cite{Qingqing_Q}. This assumption has not been considered in the previous works \cite{Bi1,Lu,Cui1,Cheng3,Ye,Kim2017}. Indeed, each WD can utilize the energy left from previous phase  to achieve a high EE in the current phase. Besides, $Q_k$ could also be the energy harvested from other sources such as wind and solar \cite{Qingqing_Q}. Therefore, the total energy consumption of all WDs can be expressed as $E_\text{total}=\sum_{\scriptstyle k= 1}^K (E_{1,k}+E_{2,k}).$

	\section{Energy-Efficient Resource Allocation}
	This section studies  the resource allocation for the considered system to maximize the system EE. We thus  jointly optimize the time/power allocations, backscattering coefficients, local computing frequencies, execution times, and RIS phase shifts. We  first define the system EE, which is the ratio of the total system throughput and the total energy consumption given by $\eta_{EE}=\frac{R_{\text{sum}}}{E_{\text{total}}}$. For simplicity and convenience,  define  vector  optimization variables: $\mathbf{t} = [t^b_1,...,t^b_K,t^a_1,...,t^a_K]$ representing the  time allocation variables; $\mathbf{p} = [P_0,p_1,...,p_K]$ being the transmit powers of the PB and each WD; $\boldsymbol{\tau} = [\tau_1,...,\tau_K]$ represents the collection of execution times; $\mathbf{f} = [f_1,...,f_K]$ represents the collection of local computing frequencies; ${\boldsymbol{\rho}} = [\rho_1,\rho_2,...,\rho_K]^T$ represents the collection of all backscattering coefficients;	${\boldsymbol{\theta}} = [\theta_1,\theta_2,...,\theta_N]^T$ represents the collection of all phase shifts. Mathematically speaking, the system EE maximization problem can be formulated as 
	\begin{subequations}
		\begin{align}
		\text{(P1)}: ~~& \underset{\mathbf{t},\boldsymbol{\rho},\mathbf{p},\boldsymbol{\tau},\mathbf{f},\boldsymbol{\theta}}{\text{maximize}} \: \: \eta_{EE}(\mathbf{t},\boldsymbol{\rho},\mathbf{p},\boldsymbol{\tau},\mathbf{f},\boldsymbol{\theta})\\
	    &\text{s.t.}\quad \Gamma^t_k+ {\Gamma}_k\geq \gamma _{\min,k },~\forall k, \label{p1-1}\\
		&\quad\quad
		{E}_{1,k} +{E}_{2,k}\leq E^t_k +Q_k,~\forall k,\label{p1-2}\\
		&\quad\quad \sum\limits_{\scriptstyle k= 1}^K t^b_k+t^o_k\leq T, 0\leq \tau_k \leq T,~\forall k,\label{p1-3}\\
		&\quad\quad 0\leq f_k \leq f_\text{max},~\forall k,\label{p1-4}\\
		&\quad\quad 0\leq \rho_k \leq 1,~\forall k,\label{p1-5}\\
		&\quad\quad P_0 \leq P_\text{max}\label{p1-6},\\
		&\quad\quad {|{\boldsymbol{\theta}}_{n,n} |=1, \forall n,}\label{p1-66}\\
		&\quad\quad p_k\geq 0, t^b_k\geq 0, t^o_k\geq 0,~\forall k\label{p1-7},
		\end{align}
	\end{subequations}
	where (\ref{p1-1}) ensures the minimum required computation task bits for each WD. (\ref{p1-2}) guarantees that the energy consumption of each WD in the uplink is limited to the harvested energy, $E^t_k$, and the primary energy, $Q_k$. In (\ref{p1-3}) and (\ref{p1-4}), $T$ and $f_\text{max}$ imply the total available transmission time for the considered time block and the maximum CPU frequency of each WD, respectively. Constraint (\ref{p1-5}) denotes the backscattering coefficient constraint. Constraint (\ref{p1-6}) limits the downlink transmit power of the PB to $P_\text{max}$, and (\ref{p1-66}) guarantees that the diagonal phase shift matrix has $N$ unit-modulus elements on its main diagonal. Finally, (\ref{p1-7}) is the non-negativity constraint on the power control and time allocation variables. In the following, we first study problem (P1) in more detail.

	\begin{Proposition}
		The maximum EE can always be obtained  when
        each WD performs local computation during each time block, i.e., $\tau^*_k= T$.
	\end{Proposition}	
	\begin{proof}
	While $(t^b_k)^*,(t^o_k)^*,\boldsymbol{\rho}^*,\mathbf{p}^*$, and $\{f_i,\tau_i\}$, $\forall k \in \mathcal{K}\backslash  k$ are fixed, we jointly optimize $f_k$ and $\tau_k$, $\forall i \in \mathcal{K}$ (the optimal computing frequency and execution time for WD $k$) to maximize the system EE. By resorting to contradiction, we prove that  the maximum  system EE is obtained when $\tau^*_k=T$. Let assume that for given $(t^b_k)^*,(t^o_k)^*,\boldsymbol{\rho}^*,\mathbf{p}^*$, and $\{f_i,\tau_i\}$, $\forall i \in \mathcal{K}\backslash  k$, the optimal solutions i.e., $\tau^*_k \leq T$ and $\{f^*_k,\tau^*_k\}$, $\forall k \in \mathcal{K}$   meets constraints $\text{(\ref{p1-1})--(\ref{p1-7})}$. Then, we construct another solution satisfying $\hat{\tau}_k=T$ and $\hat{\tau}_k(\hat{f}_k)^3={\tau}^*_k({f}^*_k)^3$. According to $\hat{\tau}_k=T> \tau^*_k$, it can be concluded that $\hat{f}_k<{f}^*_k$. Consequently, the constructed solution also meets constraints $\text{(\ref{p1-1})--(\ref{p1-7})}$. Given that $\hat{\tau}_k\hat{f}_k(\hat{f}_k)^2=\tau^*_k{f}^*_k({f}^*_k)^2$ and $\hat{f}_k<{f}^*_k$, we obtain $\hat{\tau}_k\hat{f}_k > \tau^*_k{f}^*_k$. Therefore, the constructed solution can achieve a higher system EE which results in higher throughput. This contradicts the earlier assumption that $\tau^*_k \leq T$. Thus, Proposition 1 is proven.	
	\end{proof}
	 
	Lemma 1 implies that consuming the whole available transmission time ($T$) is optimal. Indeed, if the total available time is not entirely consumed, increasing the time for both BC and AT  by the same factor keeps the EE at the same degree while promoting the total throughput. In other words, if each WD performs local computing during each time block, that achieves the maximum EE. 
	
		\begin{Proposition}
		The maximum EE can always be obtained by consuming all the available transmission time, i.e., $\sum_{\scriptstyle k= 1}^K t^b_k+t^o_k = T$.
	\end{Proposition}	
	\begin{proof}
Suppose that $\{(t^b_k)^*,(t^o_k)^*,\boldsymbol{\rho}^*,\mathbf{p}^*,{\boldsymbol{\tau}}^*_k,\mathbf{f}^*\}$ yields the maximum system EE denoted by ${EE}^*$, and satisfies $0\leq \sum_{\scriptstyle k= 1}^K t^b_k+t^o_k \leq T, \forall k$.
Subsequently, we construct another solution set defined as $\{t^b_k,t^o_k,\tilde{\boldsymbol{\rho}},\tilde{\mathbf{p}},\tilde{{\boldsymbol{\tau}}},\tilde{\mathbf{f}}\}$
which satisfies $\tilde{t}^b_k=\alpha(t^b_k)^*$, $\tilde{t}^o_k=\alpha(t^o_k)^*$ $\tilde{\boldsymbol{\rho}}=\boldsymbol{\rho}^*$, $\tilde{\mathbf{p}}=\mathbf{p}^*$, $\tilde{{\boldsymbol{\tau}}}={\boldsymbol{\tau}}^*$, and $\tilde{\mathbf{f}}={\mathbf{f}}^*$, respectively,
where $\alpha=\frac{T}{\sum_{\scriptstyle k= 1}^K (t^b_k)^*+(t^o_k)^*}\geq 1$ such that $\sum_{\scriptstyle k= 1}^K \tilde{t}^b_k+\tilde{t}^o_k  =T$. The corresponding system EE is defined as $\tilde{{EE}}$. First,
it is straightforward to verify that $\{t^b_k,t^o_k,\tilde{\boldsymbol{\rho}},\tilde{\mathbf{p}},\tilde{{\boldsymbol{\tau}}},\tilde{\mathbf{f}}\}$ still meets constraints $\text{(\ref{p1-1})--(\ref{p1-7})}$. Next, replacing $\{t^b_k,t^o_k,\tilde{\boldsymbol{\rho}},\tilde{\mathbf{p}},\tilde{{\boldsymbol{\tau}}},\tilde{\mathbf{f}}\}$ into
(P1) yields ${\tilde{EE}}={EE}^*$, which implies that the optimal system EE can always be obtained by consuming all the available transmission time, i.e., $T$. Proposition 2 is thus proved.
	\end{proof}

On the other hand, if $\sum_{\scriptstyle k= 1}^K t^b_k+t^o_k \leq T$ holds, we can increase $t^b_k$ and $t^o_k$  by the same factor such that constraint (\ref{p1-3}) is held with strict equality, i.e., $\sum_{\scriptstyle k= 1}^K t^b_k+t^o_k= T$. This is because  the scaled $t^b_k$ and $t^o_k$ also satisfy constraint (\ref{p1-3}) and reach  the same EE. Nevertheless, note  that the total  throughput  increases linearly over  $t^b_k$ and $t^o_k$ and thus  will be scaled accordingly when  $t^b_k$ and $t^o_k$ are scaled up. The key takeaway is that  the throughput can be improved without decreasing the EE but at the cost of increasing the transmission time.


	\begin{Proposition}
		For problem (P1), the maximum EE can always be obtained for $P_0=P_\text{max}$. Thus, $\mathbf{p} = [P_0,p_1,...,p_K]$ changes into $\mathbf{p} = [p_1,...,p_K]$.
	\end{Proposition}	
	\begin{proof}
Since the power transfer may not be activated because of the WDs' primary energy, we consider the subsequent two cases. First, when the power transfer is activated for the optimal solution, i.e., $(t^b_k)^*\geq 0$, then it can be determined that $P^*_0 = P_\text{max}$. Supposed that $\{\mathbf{t}^*,\boldsymbol{\rho}^*,\mathbf{p}^*,{\boldsymbol{\tau}}^*_k,\mathbf{f}^*\}$ is the optimal solution set of (P1), where $P^*_0 \leq P_\text{max}$ holds for any $P^*_0$. By denoting the optimal solution of the system EE as $EE^*$, we construct another solution set defined as $\{\tilde{\mathbf{t}},\tilde{\boldsymbol{\rho}},\tilde{\mathbf{p}},\tilde{\boldsymbol{\tau}},\tilde{\mathbf{f}}\}$ which satisfies ${\tilde{P}}_0={P}_\text{max}$, ${{p}}_k={p}^*_k$, $\tilde{{t}}^b_k=\frac{{{P}}^*_0({{t}^b_k})^*}{{\tilde{P}}_0}$, $\tilde{{t}}^o_k=({{t}^o_k})^*$, $\tilde{\boldsymbol{\rho}}=\boldsymbol{\rho}^*$, $\tilde{\boldsymbol{\tau}}=\boldsymbol{\tau}^*$, and $\tilde{\mathbf{f}}={\mathbf{f}}^*$. Next, by denoting the corresponding system EE as $\tilde{EE}$, we can conclude that $\{\tilde{\mathbf{t}},\tilde{\boldsymbol{\rho}},\tilde{\mathbf{p}},\tilde{\boldsymbol{\tau}},\tilde{\mathbf{f}}\}$ is a feasible solution set for (P1). Besides, as ${\tilde{P}}_0={P}_\text{max}\geq P^*_0$, it follows that $\tilde{t}^b_k\leq ({t}^b_k)^*$. Thus, the energy consumption satisfies the following inequality $
P_{c,k}\tilde{ t}^b_k+\frac{p_k}{\delta}\tilde{t}^o_k + p_{c,k}\tilde{t}^o_k+\epsilon_k\tilde{f}_k^3T\leq P_{c,k}(t^b_k)^*+\frac{p_k}{\delta}(t^o_k)^* + p_{c,k}(t^o_k)^*+\epsilon_k(f_k^3)^*T,~\forall k.$
From  the definition of the system EE, we  interpret that $EE^* \leq EE$ contradicts  the fact that $\{\mathbf{t}^*,\boldsymbol{\rho}^*,\mathbf{p}^*,{\boldsymbol{\tau}}^*_k,\mathbf{f}^*\}$ is the optimal solution set of (P1). On the other hand, when $(t^b_k)^*=0$ holds, then the value of the transmission power at the PB does not change the maximum system EE so that $P_0 = P_\text{max}$ is also an optimal solution.
	\end{proof}
	\section{Equivalent Optimization Problem}
	The system EE maximization problem (7a) is a MOOP. The EE is the ratio between the system throughput $R_{sum}$ and $E_{total}$ of all WDs, which are conflicting objectives. Thus, to maximize  EE, we maximize  $R_{sum}$ and  minimize   $E_{total}$  simultaneously. 
	
	To explain this process, we first present  several   MOOP concepts and terminology \cite{Miettinen}. A MOOP is  the maximization of  $F(x) = (f_1(x), f_2(x), \ldots, f_m(x) ), $ where $ x \in \mathcal S$, $\mathcal S $ is the feasible set,  and $f_i(x), i=1,\ldots, m$ are the objective functions.  The objectives in $F(x) $ often are mutually  conflicting, and  thus an increase in $f_i(x) $  may lead to a decrease of $ f_j (x)$. Thus, a single solution $x$, optimizing  all objectives simultaneously, does not exist. Instead,   Pareto optimal solutions achieve different trade-offs among the mutually conflicting objectives. A solution $x^* \in \mathcal S $ is Pareto optimal  if there does not exist another solution $ x \in \mathcal S $ such that $ f_i (x) \geq f_i(x^*)$ for all $ i=1,\ldots, m$ and $ f_j(x) > f_j(x^*)$ for at least one $j$.

    The  utopia solution  is the individual optimization of the  objectives, e.g., $z=(z_1,\ldots, z_m)$ where  $z_i^*= \max_x f_i(x),~\forall i$. Note that the utopia solution is not feasible in general due to conflicts among the objectives.  The  boundary achieved by all feasible points from the Pareto optimal set is called the Pareto optimal front (PF).  Here, we  adopt the weighted Tchebycheff method, which can yield the PF, even if the MOOP is non-convex  \cite{Mili,Mehrdad_Kolivand}. It   minimizes the maximum deviations of the objective values from their utopia values; e.g., one solves  \begin{equation}
        \min_{x}\max_{i=1,\ldots, m}\{ w_i|f_i(x)-z_i^*|\} , 
    \end{equation}  where  $w_i$, $i=1,\ldots, m, $ denotes  a set of weights.  For the problem at hand, (P1), there are two objectives ($m=2$).   With  an auxiliary optimization variable $\chi$, the epigraph representation of (P1) can be expressed as follows:
	\begin{subequations}
		\begin{align}
		\text{(P2)}:~~& \underset{\chi,\mathbf{t},\boldsymbol{\rho},\mathbf{p},\boldsymbol{\tau},\mathbf{f},\boldsymbol{\theta}}{\text{minimize}} \: \: \chi\\
		&\text{s.t.}\quad \text{(\ref{p1-1})--(\ref{p1-7})},\\
		&\quad\quad \frac{\alpha}{|R_{\text{sum,max}}|}\left(R_{\text{sum,max}}-R_{\text{sum}} \right)-\chi\leq 0, \label{p2-1}\\
		&\quad\quad \frac{\beta}{|E_{\text{total,max}}|}\left(E_{\text{total}}-E_{\text{total,max}} \right)-\chi\leq 0. \label{p2-2}
		\end{align}
	\end{subequations}
	Note that $R_{\text{sum,max}}$ and $E_{\text{total,max}}$ are the utopian objective points. $|R_{\text{sum,max}}|$ as well as $|E_{\text{total,max}}|$ denote the normalization factors. In addition, $0\leq\alpha\leq 1$ and $0\leq\beta\leq 1$ ( $\alpha+\beta=1$) are the weighting coefficients that reveal the importance of different objectives. Furthermore, it  should also be noted that (P2) is equivalent to the problem of maximizing throughput when $ \alpha = $ 1 and $ \beta = 0 $ and is equivalent to the problem of minimizing energy when $ \alpha = $ 0 and $\beta =1$, meaning that both MOOP and SOOP have the same optimal solution.
	
{\begin{Proposition}
  The solution to the SOOP given in (P2) is a special case of the EE maximization problem solution in (P1).
\end{Proposition}}
	
{\begin{proof}
To sketch the proof, let us first formulate  general fractional programming as follows:
\begin{equation}\label{cond_00}
\underset{x}{\text{min}}\: \:\frac{f(x)}{g(x)},\quad\text{s.t.}\quad x\in \mathcal{X},
\end{equation}where $g(x)>0$, $\forall x\in \mathcal{X}$ and $\mathcal{X}\in \mathbb{R}^n$ is a nonempty compact. In addition, $f(x)$ and $g(x)$ are both continuous real-valued  functions of $x\in \mathcal{X}$. To obtain the optimal solution, let us represent the following function:
\begin{equation}\label{cond_0}
Y(\varpi^{*})=\min_{x}\: \: \left\{f(x)-\varpi^{*}g(x),\quad \text{s.t.}\quad x\in \mathcal{X}\right\}
\end{equation}
as the minimum value of $f(x)-\varpi g(x)$ with each constant $\varpi^{*}$. According to the Dinklebach method \cite{Dinkelbach}, it is proven  that
\begin{equation}\label{cond_1}
\varpi^{*}=\frac{f(x^{*})}{g(x^{*})}=\min_{x}\: \: \left\{\frac{f(x)}{g(x)},\quad\text{s.t.}\quad x\in \mathcal{X}\right\}
\end{equation}if and only if
\begin{equation}\label{cond_2}
Y(\varpi^{*})\!=\!Y(\varpi^{*},x^{*})\!=\!\min_{x}\: \: \left\{f(x)-\varpi^{*}g(x),\:\:\text{s.t.}\:\: x\in \mathcal{X}\right\}\!=\!0.
\end{equation} 
Consequently, from (\ref{cond_1}) and (\ref{cond_2}), it can be shown 
the optimal solution, $x^{*}$, of (\ref{cond_00}) is equivalent to the solution obtained in  (\ref{cond_0}) when $\varpi=\varpi^{*}$, where  $\varpi^{*}$ indicates the minimum value of (\ref{cond_00}). Next, let us formulate the  MOOP given as below:
\begin{subequations}
		\begin{align}
		 & \underset{x}{\text{min}} \: \: f(x) \label{10aa}\\
		 & \underset{x}{\text{max}} \: \: g(x) \label{10bb}\\
		&\text{s.t.}\quad  x>0,
		\end{align}
\end{subequations}
where $f(x)$ and $g(x)$ represent the numerator and denominator of the fractional optimization problem  (\ref{cond_00}), respectively. By incorporating the competing objective functions (\ref{10aa}) and (\ref{10bb}) into a SOOP via the Tchebycheff method, we transform  the objective functions in the MOOP into the following SOOP:
\begin{subequations}
		\begin{align}
		& \underset{x}{\text{min}} \: \: \chi \label{11aa}\\
		&\quad\quad \frac{\alpha}{|f_{\max}(x)|}\left(f_{\max}(x)-f(x) \right)-\chi\leq 0, \nonumber \\
		&\quad\quad \frac{\beta}{|g_{\max}(x)|}\left(g_{\max}(x) -g(x) \right)-\chi\leq 0,\nonumber \\
		&\text{s.t.}\quad x>0,
		\end{align}
\end{subequations}
Based on equations (\ref{11aa}) and (\ref{cond_1}), it can be easily  concluded that the optimal set of (\ref{11aa}) includes the solution of (\ref{cond_1}), and the values $\alpha$ and $\beta$ also provide a solution for the fractional programming problem.
\end{proof}	}

	\section{Solution to the Optimization Problem}
	For the first SOOP (throughput maximization), we exploit the AO  method to divide the main problem into two sub-problems. For the first one, we derive closed-form  resource allocations. For the second one, we apply the SDR  technique, a difference of convex functions (DC) programming, MM  techniques, along with a penalty function for enforcing a rank-one solution to design RIS phase shifts. For the second SOOP (energy consumption minimization), we only derive closed-form  resource allocations since the energy consumption is independent of RIS phase shifts.
	
	Define the following auxiliary variables, $s_k=\rho_k t^b_k$ and $z_k=p_kt^o_k$, $\forall k $.  The equivalent form of (P1) can then be recast as follows:
	\begin{subequations}
		\begin{align}
		\text{(P3)}: ~~& \underset{\chi,\mathbf{t},\mathbf{s},\mathbf{z},\mathbf{f}}{\text{minimize}} \: \: \chi\\
		&\text{s.t.}\quad \text{(\ref{p1-3})}, \text{(\ref{p1-4}), (\ref{p1-5})},\label{p3}\\
		&\quad\quad \bar{R}_{\text{sum}} \geq R_{\text{sum,max}} -\chi\frac{|R_{\text{sum,max}}|}{\alpha}, \label{p4-1}\\
		&\quad\quad \bar{E}_{\text{total}} \leq E_{\text{total,max}} +\chi\frac{|E_{\text{total,max}}|}{\beta}, \label{p4-2}\\
		&\quad\quad \bar{\Gamma}^t_k+ {\bar{\Gamma}}_k\geq \gamma_{\min,k },~\forall k, \label{p4-3}\\
		&\quad\quad
		{E}_{1,k} 	+	\bar{E}_{2,k} \leq \bar{E}^t_k +Q_k,~\forall k,\label{p4-4}\\
		&\quad\quad 0\leq s_k \leq t^b_k,~\forall k,\label{p4-5}\\
		&\quad\quad z_k\geq 0, t^b_k\geq 0, t^o_k\geq 0,~\forall k,\label{p4-6}
		\end{align}
	\end{subequations}
	where $\mathbf{s} = [s_1,...,s_K]$ and $\mathbf{z} = [z_1,...,z_K]$ denote 
	the collocation of auxiliary variables. In (P3), we have $\bar{R}_{\text{sum}}=\sum_{ k= 1}^K \bar{\Gamma}^b_k + \bar{\Gamma}^o_k+ {\bar{\Gamma}}_k$, $\bar{E}_\text{total}= \sum_{ k= 1}^K {E}_{1,k} +\bar{E}_{2,k} $, and $\bar{E}^t_k =\bar{E}^b_k+ \sum_{ i= 1, i \ne k}^K\bar{P}^b_kt^b_i$, where
	\begin{subequations}
		\begin{align}
		&\bar{\Gamma}^b_k = t^b_k\log_2(1+\frac{\zeta s_k P_\text{max}|h_k|^2|g_k|^2}{t^b_k \sigma^2}),\nonumber\\& \bar{\Gamma}^o_k =t^o_k\log_2(1+\frac{z_k|h_k|^2}{t^o_k\sigma^2}),~\forall k,\\
		& \bar{P}^b_k=\frac{a_kP_\text{max}|g_k|^2+b_k}{P_\text{max}|g_k|^2+c_k}-\frac{b_k}{c_k},\nonumber\\&\bar{\Gamma}_k =\frac{Tf_k}{C_\text{cpu}},~\forall k,\\
		&\bar{E}_{2,k} = \frac{z_k}{\delta} + p_{c,k}t^o_k+\epsilon_kf_k^3T,\nonumber\\ 
		& \bar{E}^b_k = \left(\frac{a_k(1-\frac{s_k}{t^b_k})P_\text{max}|g_k|^2+b_k}{(1-\frac{s_k}{t^b_k})P_\text{max}|g_k|^2+c_k}-\frac{b_k}{c_k}\right)t^b_k,~\forall k.
		\end{align}
	\end{subequations}
	However, due to the non-linear EH constraint (\ref{p4-4}),  (P3) is not  a convex problem. To deal with it, we resort to Proposition 1  \cite{Lu},  which is  that the right-hand side of (\ref{p4-4}) is a concave function. This can be proven with the  first and second-order derivatives  of (\ref{p4-4}). Next, we first solve corresponding SOOPs  to determine $R_{\text{sum,max}}$ and $E_{\text{total,max}}$. 
	
		\subsection{Proposed Solution for Throughput Maximization}
 
		\subsubsection{Sub-problem 1: Optimizing $\{\mathbf{t},\boldsymbol{\rho},\mathbf{p},\boldsymbol{\tau},\mathbf{f}\}$}
	By setting $\beta=0$, we get the throughput maximization problem to obtain the optimal value of $R_{\text{sum,max}}$ as 
	\begin{subequations}
		\begin{align}
		\text{(P4)}: ~~& \underset{\mathbf{t},\mathbf{s},\mathbf{z},\mathbf{f}}{\text{maximize}} \: \:\bar{R}_{\text{sum}}\\
		&\text{s.t.}\quad \text{(\ref{p3}), (\ref{p4-3})--(\ref{p4-6})}.
		\end{align}
	\end{subequations}
	Since (P4) is a convex optimization problem and satisfies Slater’s condition, we can apply the Lagrange duality method to derive closed-form expressions. 
	The following theorem provides its  optimal solutions in closed-form expressions.
	\begin{theorem}
		Given the non-negative Lagrange multipliers,
		i.e., $\{\lambda,\mu_k,\omega_k,\Omega_k,\nu_k\}$, the maximizer of Lagrangian function is given by
\begin{subequations}\label{lagrang1}
			{\small		\begin{align}
			&p^*_k=\left[\frac{(1+\Omega_k)\delta}{\nu_k\ln2}-\frac{\sigma^2}{|h_k|^2}\right]^+,\; f^*_k=\left[\sqrt{\frac{(1 +\Omega)T-\mu_kC_\text{cpu}}{3\epsilon_k\nu_kTC_\text{cpu}} }\right]^+,\\
			&	\rho^*_k=
			\begin{cases}
			\left[\frac{Y_k-\sqrt{Y^2_k-4X_kZ_k}}{2X_k}\right]^+, & \omega_k=0,\\
			0, & \omega_k>0,
			\end{cases}
			\end{align}}
		\end{subequations}
		where $X_k = \frac{\zeta (1+\Omega_k) P^3_{\max}|g_k|^6|h_k|^2}{\ln2},$
		$Y_k = 2X_k+\frac{2\zeta c_k(1+\Omega_k) P_{\max}|g_k|^4|h_k|^2}{\ln2}+\zeta \nu_k(a_kc_k-b_k)P_{\max}|g_k|^2|h_k|^2,$
		$Z_k = X_k+\frac{\zeta c^2_k(1+\Omega_k) P_{\max}|g_k|^2|h_k|^2}{\ln2}+\frac{2\zeta c_k(1+\Omega_k) P^2_{\max}|g_k|^4|h_k|^2}{\ln2}-\nu_k(a_kc_k-b_k)\sigma^2.$
	\end{theorem}
	\begin{proof}
	The Lagrangian function of problem (P4) is given by (\ref{lagrang111}),
\begin{figure*}[t]
	\centering
	\begin{align}
	&\mathcal{L}(\lambda,\mu_k,\omega_k,\Omega_k,\nu_k)=\lambda\bigg[ T-\sum_{k\in\mathcal{K}}({t^b_k+t^o_k})\bigg] +\sum_{k\in\mathcal{K}}\mu_k[f_\text{max}-f_k]+\sum_{k\in\mathcal{K}}\omega_k[t^b_k-s_k]\nonumber\\&\sum_{k\in\mathcal{K}}(1+\Omega_k)\bigg[ t^b_k\log_2(1+\frac{\zeta	 s_kP_\text{max}|h_k|^2|g_k|^2}{t^b_k\sigma^2})+t^o_k\log_2(1+\frac{z_k|h_k|^2}{t^o_k\sigma^2})+\frac{Tf_k}{C_\text{cpu}}\bigg]-\sum_{k\in\mathcal{K}}\Omega_k\gamma_{\min,k} \nonumber\\& -\sum_{k\in\mathcal{K}}\nu_k\bigg[ \frac{z_k}{\delta}+ p_{c,k}t^o_k+\epsilon_kf_k^3T+P_{c,k}t^b_k-\left(\frac{a_k(1-\frac{s_k}{t^b_k})P_\text{max}|g_k|^2+b_k}{(1-\frac{s_k}{t^b_k})P_\text{max}|g_k|^2+c_k}-\frac{b_k}{c_k}\right)t^b_k-Q_k -\sum_{i\in\mathcal{K},i\neq k}\bar{P}^b_k t^b_i\bigg],\label{lagrang111}
	\end{align}
	\medskip
	\hrule
\end{figure*}
where $\lambda$, $\mu_k$, $\omega_k$, $\Omega_k$, and $\nu_k$ are the non-negative Lagrange multipliers associated with the constraints of problem (P4). By taking the partial derivative of $\mathcal{L}(\lambda,\mu_k,\omega_k,\Omega_k,\nu_k)$ with respect to the optimization variables in (P4), we obtain
\begin{align}
\frac{\partial\mathcal{L}}{\partial{s_k}}
&	=\frac{\zeta t^b_kP_\text{max}|h_k|^2|g_k|^2( 1+\Omega_k)}{(t^b_k\sigma^2+\zeta s_kP_\text{max}|h_k|^2|g_k|^2)\ln2}\nonumber\\&-\frac{(a_kc_k-b_k)(t^b_k)^2( \nu_k)}{\left( (t^b_k-s_k)P_\text{max}|g_k|^2+c_kt^b_k\right) ^2}-\omega_k,~\forall k,\label{sss_k}\\
\frac{\partial\mathcal{L}}{{\partial{z_k}}}
&	=\frac{t^o_k|h_k|^2( 1+\Omega_k) }{(t^o_k\sigma^2+z_k|h_k|^2)\ln2}-\frac{\nu_k}{\delta},~\forall k,\\
\frac{\partial\mathcal{L}}{\partial{f_k}}
&	=\left( \frac{1+\Omega_k}{C_\text{cpu}}-3\nu_k\epsilon_kf^2_k\right) T-\mu_k,~\forall k.\label{sss_k2}
\end{align}	
By setting $\frac{\partial\mathcal{L}}{\partial{f_k}}$ and $\frac{\partial\mathcal{L}}{{\partial{z_k}}}=0$ as zero, the optimal CPU frequency and transmit power of WD $k$ can be calculated, respectively, as
\begin{align}
&\frac{\partial\mathcal{L}}{\partial{f_k}}=0\Rightarrow f^*_k=\left[\sqrt{\frac{(1 +\Omega)T-\mu_kC_\text{cpu}}{3\epsilon_k\nu_kTC_\text{cpu}} }\right]^+,\nonumber\\& \frac{\partial\mathcal{L}}{{\partial{z_k}}}=0\Rightarrow z^*_k=\left[\frac{(1+\Omega_k)t^o_k\delta}{\nu_k\ln2}-\frac{t^o_k\sigma^2}{|h_k|^2}\right]^+,~\forall k.
\end{align} 
Next, by using $z_k=p_kt^o_k$, the optimal transmit power can be calculated as $
p^*_k=\left[\frac{(1+\Omega_k)\delta}{\nu_k\ln2}-\frac{\sigma^2}{|h_k|^2}\right]^+.$
To derive the closed-form expression for $\rho_k$, we study (\ref{sss_k}). In particular, if $s_k = t^b_k$, the optimal backscattering coefficient is $\rho_k^*=1$. However, if $s_k \leq t^b_k$ holds, $\omega_k=0$ must be satisfied according to the Karush-Kuhn-Tucker (KKT) conditions. In this case, the optimal backscattering coefficient should satisfy: $X_k(\rho^*_k)^2- Y_k\rho^*_k+ Z_k=0$, where
$
X_k = \frac{\zeta (1+\Omega_k) P^3_{\max}|g_k|^6|h_k|^2}{\ln2},$
$Y_k = 2X_k+\frac{2\zeta c_k(1+\Omega_k) P_{\max}|g_k|^4|h_k|^2}{\ln2}+\zeta \nu_k(a_kc_k-b_k)P_{\max}|g_k|^2|h_k|^2,$ and
$	Z_k = X_k+\frac{\zeta c^2_k(1+\Omega_k) P_{\max}|g_k|^2|h_k|^2}{\ln2}+\frac{2\zeta c_k(1+\Omega_k) P^2_{\max}|g_k|^4|h_k|^2}{\ln2}-\nu_k(a_kc_k-b_k)\sigma^2,~\forall k.$
Since $0 \leq \rho^*_k\leq 1$ and 	$\frac{Y_k+\sqrt{Y^2_k-4X_kZ_k}}{2X_k}>1$, we have
\begin{align}
\rho^*_k=
\begin{cases}
\left[\frac{Y_k-\sqrt{Y^2_k-4X_kZ_k}}{2X_k}\right]^+, & \omega_k=0,\\
0, & \omega_k>0.
\end{cases}
\end{align} 
Note that  closed-form expressions for BC time, $t^b_k$,  and AT time, $t^o_k,$ appear intractable as the Lagrangian functions, i.e., $	\frac{\partial\mathcal{L}}{\partial{t^b_k}}$ and $	\frac{\partial\mathcal{L}}{\partial{t^o_k}}$ are linear functions with respect to  $t^b_k$ and $t^o_k$, respectively. To obtain the optimal values of these time allocations, a linear optimization package, e.g., the simplex method, can be applied, which usually finds a solution quickly.  
{Consequently, the dual variables are updated according to the gradient method with given optimization variables \cite{Boyd}. The updated dual variables are given by }
{\begin{align}
      \lambda(t+1)&=\left[ \lambda(t)- \varkappa(t) \left(  \sum_{k\in\mathcal{K}}(T-{t^b_k+t^o_k}) \right)
     \right]^+,~\forall t,\\
      \mu_k(t+1)&=\left[  \mu_k(t)- \varkappa(t) \left(  f_\text{max}-f_k\right)
     \right]^+,~\forall t, k,\\
       \omega_k(t+1)&=\left[  \omega_k(t)- \varkappa(t) \left(t^b_k-s_k\right)
     \right]^+,~\forall t, k,\\
     \nu_k (t+1)&=\bigg[   \nu_k(t)- \varkappa(t)  \bigg(Q_k -\frac{z_k}{\delta}- p_{c,k}t^o_k-\epsilon_kf_k^3T\nonumber\\
     &+ \bigg(\frac{a_k(1-\frac{s_k}{t^b_k})P_\text{max}|g_k|^2+b_k}{(1-\frac{s_k}{t^b_k})P_\text{max}|g_k|^2+c_k}-\frac{b_k}{c_k} \bigg)t^b_k-P_{c,k}t^b_k \nonumber\\
     &+\sum_{i\in\mathcal{K},i\neq k}\bar{P}^b_k t^b_i \bigg)  \bigg]^+,~\forall t, k,\\
     \Omega_k(t+1)&=\bigg[  \Omega_k(t)- \varkappa(t) \bigg( t^o_k\log_2(1+\frac{z_k|h_k|^2}{t^o_k\sigma^2})-\gamma_{\min,k}\nonumber\\&+ t^b_k\log_2(\!1\!+\!\frac{\zeta	 s_kP_\text{max}|h_k|^2|g_k|^2}{t^b_k\sigma^2}\!)\!+\!\frac{Tf_k}{C_\text{cpu}}\bigg)
    \bigg]^+\!\!,~\forall t, k,
\end{align}} {where $ \varkappa(t) \geq 0$ is a positive step size at iteration $t$. Then, the dual variables  $\{ \lambda,\mu_k,\omega_k,\Omega_k,\nu_k\}$  and the primal variables $\{\mathbf{t},\boldsymbol{\rho},\mathbf{p},\boldsymbol{\tau},\mathbf{f}\}$ are optimized iteratively.} 
	\end{proof}

	\subsubsection{Sub-problem 2: Optimizing $\boldsymbol{\Theta}$}
	By using the  optimal solutions of  sub-problem 1, i.e., $\{\mathbf{t}^*,\boldsymbol{\rho}^*,\mathbf{p}^*,\boldsymbol{\tau}^*,\mathbf{f}^*\}$\footnote{Note that in the rest of the paper, we consider the optimal solutions $\{\mathbf{t}^*,\boldsymbol{\rho}^*,\mathbf{p}^*,\boldsymbol{\tau}^*,\mathbf{f}^*\}$ as $\{\mathbf{t},\boldsymbol{\rho},\mathbf{p},\boldsymbol{\tau},\mathbf{f}\}$}, we propose a suboptimal algorithm to determine $\boldsymbol{\Theta}$. Since the denominator of $\eta_{EE}$, i.e., $E_{\text{total}}$ is not a function of $\boldsymbol{\Theta}$, (P1) can be equivalently written as
	\begin{subequations}
		\begin{align}
		\text{(P5)}: ~~& \underset{\boldsymbol{\Theta}}{\text{maximize}} \: \: \bar{R}_{\text{sum}} \\
		&\text{s.t.}\quad \text{(\ref{p1-1}), (\ref{p1-2})}, \text{(\ref{p1-66})}.
		\end{align}
	\end{subequations}
	By defining $\Xi_k=t^b_k\log_2 (t^b_k\sigma^2)+t^o_k\log_2 (t^o_k\sigma^2)-\frac{Tf_k}{C\text{cpu}}$, the objective function in (P5) can be represented as $\bar{R}_{\text{sum}}=\sum_{\scriptstyle k= 1}^K t^b_k\log_2(t^b_k\sigma^2+\zeta s_k P_\text{max}|h_k|^2|g_k|^2)+t^o_k\log_2(t^o_k\sigma^2+z_k|h_k|^2)-\Xi_k.\label{41a}$ Next, we exploit the SDR and MM techniques to determine the RIS phase shifts. Let us rewrite $|h_k|^2$ and $|g_k|^2$ as follows:
	\begin{subequations}
		\begin{align}
		|h_k|^2&={\boldsymbol{\theta}} ^H\boldsymbol{\phi}_{\text{UIM},k}\boldsymbol{\phi}^H_\text{UIM}{\boldsymbol{\theta}}+{\boldsymbol{\theta}} ^H\boldsymbol{\phi}_{\text{UIM},k}h_{\text{UM},k}^h\nonumber\\&+h_{\text{UM},k}\boldsymbol{\phi}^H_\text{UIM}{\boldsymbol{\theta}}+|h_{\text{UM},k}|^2,~\forall k,\label{45}\\
		|g_k|^2&={\boldsymbol{\theta}} ^H\boldsymbol{\phi}_{\text{PIU},k}\boldsymbol{\phi}^H_\text{PIU}{\boldsymbol{\theta}}+{\boldsymbol{\theta}} ^H\boldsymbol{\phi}_{\text{PIU},k}g_{\text{PU},k}^h\nonumber\\&+g_{\text{PU},k}\boldsymbol{\phi}^H_\text{PIU}{\boldsymbol{\theta}}+|g_{\text{PU},k}|^2,~\forall k, \label{46}
		\end{align}
	\end{subequations}
	where $\boldsymbol{\phi}_{\text{UIM},k}\triangleq\text{diag}(\mathbf{h}_{\text{UI},k}^H)\mathbf{h}_{\text{IM}}$ and $\boldsymbol{\phi}_{\text{PIU},k}\triangleq\text{diag}(\mathbf{g}_{\text{PI}}^H)\mathbf{g}_{\text{IU},k}$. Furthermore, we represent a matrix form for equations (\ref{45}) and (\ref{46}) as	\begin{align}	\label{45b}
	&|h_k|^2=\hat{{\boldsymbol{\theta}}}^H \boldsymbol{S}_k\hat{{\boldsymbol{\theta}}}+|h_{\text{UM},k}|^2,~\forall k, \nonumber\\
	&|g_k|^2=\hat{{\boldsymbol{\theta}}}^H \boldsymbol{R}_k\hat{{\boldsymbol{\theta}}}+|g_{\text{PU},k}|^2,~\forall k,
	\end{align}
	respectively, where $\hat{{\boldsymbol{\theta}}}=\left[ {\begin{array}{*{20}{c}}{{\boldsymbol{\theta}}}\\{{\mathbf{1}}}\end{array}} \right]$ and
	\begin{align}
	&\boldsymbol{S}_k=\left[ {\begin{array}{*{20}{c}}{{\boldsymbol{\phi}_{\text{UIM},k}\boldsymbol{\phi}^H_{\text{UIM},k}}}&\boldsymbol{\phi}_{\text{UIM},k}{h}^h_{\text{UM},k}\\{{{h}_{\text{UM},k}\boldsymbol{\phi}^H_{\text{UIM},k}}}&{{0}}\end{array}} \right],\nonumber\\&  \boldsymbol{R}_k=\left[ {\begin{array}{*{20}{c}}{{\boldsymbol{\phi}_{\text{PIU},k}\boldsymbol{\phi}^H_{\text{PIU},k}}}&\boldsymbol{\phi}_{\text{PIU},k}{g}^h_{\text{PU},k}\\{{{g}_{\text{PU},k}\boldsymbol{\phi}^H_{\text{PIU},k}}}&{{0}}\end{array}} \right].
	\end{align}
	By defining $\boldsymbol{\Phi}\triangleq \hat{\boldsymbol{\theta}}\hat{\boldsymbol{\theta}}^H$ and applying the identities $\hat{{\boldsymbol{\theta}}}^H \boldsymbol{S}_k\hat{{\boldsymbol{\theta}}}=\text{Tr}( \boldsymbol{S}_k\boldsymbol{\Phi})$ and $\hat{{\boldsymbol{\theta}}}^H \boldsymbol{R}_k\hat{{\boldsymbol{\theta}}}=\text{Tr}( \boldsymbol{R}_k\boldsymbol{\Phi})$, we rewrite (\ref{45b}) as a function of $\boldsymbol{\Phi}$ which then yields $|h_k|^2=\text{Tr}( \boldsymbol{S}_k\boldsymbol{\Phi})+|h_{\text{UM},k}|^2$ and $|g_k|^2=\text{Tr}( \boldsymbol{R}_k\boldsymbol{\Phi})+|g_{\text{PU},k}|^2$. In particular, by introducing BC, two quadratic terms, i.e., $|h_k|^2|g_k|^2$ and $|g_k|^2|g_k|^2$ appear in the objective function and constraint (\ref{p1-1}), respectively, which substantially contrasts with those in existing works \cite{Zhong,Bi1,Cui1,Chae,Cheng3,Zhang,Shanfeng,Tellambura,Park,Ye}. Accordingly, we first rewrite these products as follows:
	\begin{align}\label{47}
	|h_k|^2|g_k|^2&=\left( \hat{{\boldsymbol{\theta}}}^H \boldsymbol{S}_k\hat{{\boldsymbol{\theta}}}+|h_{\text{UM},k}|^2\right) \left( \hat{{\boldsymbol{\theta}}}^H \boldsymbol{R}_k\hat{{\boldsymbol{\theta}}}+|g_{\text{PU},k}|^2\right)\nonumber\\& =\hat{{\boldsymbol{\theta}}}^H \boldsymbol{S}_k\hat{{\boldsymbol{\theta}}} \hat{{\boldsymbol{\theta}}}^H \boldsymbol{R}_k\hat{{\boldsymbol{\theta}}}+|h_{\text{UM},k}|^2\hat{{\boldsymbol{\theta}}}^H \boldsymbol{R}_k\hat{{\boldsymbol{\theta}}}\nonumber\\&+|g_{\text{PU},k}|^2\hat{{\boldsymbol{\theta}}}^H \boldsymbol{S}_k\hat{{\boldsymbol{\theta}}}+|h_{\text{UM},k}|^2|g_{\text{PU},k}|^2,~\forall k,\\
	|g_k|^2|g_k|^2 &=\left( \hat{{\boldsymbol{\theta}}}^H \boldsymbol{R}_k\hat{{\boldsymbol{\theta}}}+|g_{\text{PU},k}|^2\right)\left( \hat{{\boldsymbol{\theta}}}^H \boldsymbol{R}_k\hat{{\boldsymbol{\theta}}}+|g_{\text{PU},k}|^2\right)\label{48} \nonumber \\&=\hat{{\boldsymbol{\theta}}}^H \boldsymbol{R}_k\hat{{\boldsymbol{\theta}}} \hat{{\boldsymbol{\theta}}}^H \boldsymbol{R}_k\hat{{\boldsymbol{\theta}}}+2|g_{\text{PU},k}|^2\hat{{\boldsymbol{\theta}}}^H \boldsymbol{R}_k\hat{{\boldsymbol{\theta}}}\nonumber\\&+|g_{\text{PU},k}|^2|g_{\text{PU},k}|^2,~\forall k,
	\end{align}
	which are quadratic polynomial in $\hat{\boldsymbol{\theta}}$ and non-convex functions.
	
To address this issue, we use the MM technique where a convex minorizing function is obtained in each iteration via SDR.
	A minorizer to a function $f(\mathbf{y}):\mathbb{C}^N\rightarrow\mathbb{R}$ is constructed with bounded curvature to take the second-order Taylor expansion as $f(\mathbf{y})\geq f(\mathbf{y}_0)+\text{Re}\left\lbrace\nabla f(\mathbf{y}_0)^H(\mathbf{y}-\mathbf{y}_0)\right\rbrace -\frac{l }{2}\|\mathbf{y}-\mathbf{y}_0\|^2,$ where $\mathbf{y}_0\in \mathbb{C}^N$ is any point, and $l$ denotes the maximum curvature of $f(\mathbf{y})$ \cite{Palomar}. Accordingly, we obtain lower bounds for (\ref{47}) and (\ref{48}) given in (\ref{49}) and (\ref{355}), respectively.
			\begin{figure*}[t]
		\centering
	\begin{align}\label{49}
	|h_k|^2|g_k|^2&\geq \hat{{\boldsymbol{\theta}_0}}^H \boldsymbol{S}_k\hat{{\boldsymbol{\theta}_0}} \hat{{\boldsymbol{\theta}_0}}^H \boldsymbol{R}_k\hat{{\boldsymbol{\theta}_0}}+|h_{\text{UM},k}|^2\hat{{\boldsymbol{\theta}_0}}^H \boldsymbol{R}_k\hat{{\boldsymbol{\theta}_0}}+|g_{\text{PU},k}|^2\hat{{\boldsymbol{\theta}_0}}^H \boldsymbol{S}_k\hat{{\boldsymbol{\theta}_0}}+|h_{\text{UM},k}|^2|g_{\text{PU},k}|^2\nonumber\\
	&\hat{{\boldsymbol{\theta}_0}}^H\boldsymbol{B}_k(\hat{{\boldsymbol{\theta}}}-\hat{{\boldsymbol{\theta}_0}})+(\hat{{\boldsymbol{\theta}}}-\hat{{\boldsymbol{\theta}_0}})^H\boldsymbol{B}_k\hat{{\boldsymbol{\theta}_0}}-\frac{l}{2}(\hat{{\boldsymbol{\theta}}}^H\hat{{\boldsymbol{\theta}}}-\hat{{\boldsymbol{\theta}}}^H\hat{{\boldsymbol{\theta}_0}}-\hat{{\boldsymbol{\theta}}_0}^H\hat{{\boldsymbol{\theta}}}+\|\hat{{\boldsymbol{\theta}_0}}\|^2)\nonumber\\
	&=-\frac{l}{2}(\hat{{\boldsymbol{\theta}}}^H\hat{{\boldsymbol{\theta}}}-\hat{{\boldsymbol{\theta}}}^H\hat{{\boldsymbol{\theta}_0}}-\hat{{\boldsymbol{\theta}}_0}^H\hat{{\boldsymbol{\theta}}}+\|\hat{{\boldsymbol{\theta}_0}}\|^2)+\hat{{\boldsymbol{\theta}_0}}^H\boldsymbol{B}_k\hat{{\boldsymbol{\theta}}}+\hat{{\boldsymbol{\theta}}}^H\boldsymbol{B}_k\hat{{\boldsymbol{\theta}}_0}+\kappa_1\nonumber\\
	&=-\frac{l}{2}\left( \hat{{\boldsymbol{\theta}}}^H\mathbf{I}\hat{{\boldsymbol{\theta}}}+\hat{{\boldsymbol{\theta}}}^H\left( -\frac{2}{l}\boldsymbol{B}_k\hat{{\boldsymbol{\theta}}}_0-\mathbf{I}\hat{{\boldsymbol{\theta}}}_0\right) +\left(-\frac{2}{l}\boldsymbol{B}_k\hat{{\boldsymbol{\theta}}}_0-\mathbf{I}\hat{{\boldsymbol{\theta}}}_0 \right)^H\hat{{\boldsymbol{\theta}}} \right) +\kappa_1,~\forall k,\\
	|g_k|^2|g_k|^2&\geq-\frac{l}{2}\left( \hat{{\boldsymbol{\theta}}}^H\mathbf{I}\hat{{\boldsymbol{\theta}}}+\hat{{\boldsymbol{\theta}}}^H\left( -\frac{2}{l}\boldsymbol{C}_k\hat{{\boldsymbol{\theta}}}_0-\mathbf{I}\hat{{\boldsymbol{\theta}}}_0\right) +\left(-\frac{2}{l}\boldsymbol{C}_k\hat{{\boldsymbol{\theta}}}_0-\mathbf{I}\hat{{\boldsymbol{\theta}}}_0 \right)^H\hat{{\boldsymbol{\theta}}} \right) +\kappa_2,~\forall k.\label{355}
	\end{align}
			\medskip
		\hrule
	\end{figure*}
	where
	$\boldsymbol{B}_k\triangleq \boldsymbol{R}_k\hat{{\boldsymbol{\theta}_0}}\hat{{\boldsymbol{\theta}_0}}^H\boldsymbol{S}_k+\boldsymbol{S}_k\hat{{\boldsymbol{\theta}_0}}\hat{{\boldsymbol{\theta}_0}}^H\boldsymbol{R}_k+ |h_{\text{UM},k}|^2\boldsymbol{R}_k+ |g_{\text{PU},k}|^2\boldsymbol{S}_k$ and $\boldsymbol{C}_k\triangleq 2\boldsymbol{R}_k\hat{{\boldsymbol{\theta}_0}}\hat{{\boldsymbol{\theta}_0}}^H\boldsymbol{R}_k+2 |g_{\text{PU},k}|^2\boldsymbol{R}_k$ are
	Hermitian matrices. In addition, we have $\kappa_1=\hat{{\boldsymbol{\theta}_0}}^H (\boldsymbol{S}_k\hat{{\boldsymbol{\theta}_0}}\hat{{\boldsymbol{\theta}^H_0}} \boldsymbol{R}_k+|h_{\text{UM},k}|^2 \boldsymbol{R}_k+|g_{\text{PU},k}|^2 \boldsymbol{S}_k-2\boldsymbol{B}_k)\hat{{\boldsymbol{\theta}_0}}+|h_{\text{UM},k}|^2|g_{\text{PU},k}|^2$ and $\kappa_2=\hat{{\boldsymbol{\theta}_0}}^H (\boldsymbol{R}_k\hat{{\boldsymbol{\theta}_0}}\hat{{\boldsymbol{\theta}^H_0}} \boldsymbol{R}_k+|g_{\text{PU},k}|^2 \boldsymbol{R}_k-2\boldsymbol{C}_k)\hat{{\boldsymbol{\theta}_0}}+|g_{\text{PU},k}|^4$. In particular, equations (\ref{49}) and (\ref{355}) can be expressed as $\tilde{{\boldsymbol{\theta}}}^H\mathbf{T}_k\tilde{{\boldsymbol{\theta}}}$ and $\tilde{{\boldsymbol{\theta}}}^H\mathbf{U}\tilde{{\boldsymbol{\theta}}}$, respectively, where $\tilde{{\boldsymbol{\theta}}}=\left[ {\begin{array}{*{20}{c}}{{\hat{\boldsymbol{\theta}}}}\\{{\mathbf{1}}}\end{array}} \right]$ and
	\begin{align}\label{32}
	&\boldsymbol{T}_k=-\left[ {\begin{array}{*{20}{c}}{\mathbf{I}}&-\frac{2}{l}\boldsymbol{B}_k\hat{{\boldsymbol{\theta}}}_0-\mathbf{I}\hat{{\boldsymbol{\theta}}}_0\\\left(-\frac{2}{l}\boldsymbol{B}_k\hat{{\boldsymbol{\theta}}}_0-\mathbf{I}\hat{{\boldsymbol{\theta}}}_0 \right) ^H&{{0}}\end{array}} \right], \\	
	&\boldsymbol{U}_k=-\left[ {\begin{array}{*{20}{c}}{\mathbf{I}}&-\frac{2}{l}\boldsymbol{C}_k\hat{{\boldsymbol{\theta}}}_0-\mathbf{I}\hat{{\boldsymbol{\theta}}}_0\\\left(-\frac{2}{l}\boldsymbol{C}_k\hat{{\boldsymbol{\theta}}}_0-\mathbf{I}\hat{{\boldsymbol{\theta}}}_0 \right) ^H&{{0}}\end{array}} \right]\label{332}.
	\end{align}
	Consequently, by letting $\tilde{\boldsymbol{\Phi}}\triangleq \tilde{\boldsymbol{\theta}}\tilde{\boldsymbol{\theta}}^H$, the objective function in (P5) and constraint (\ref{p1-1}) can be rewritten as given in (\ref{wer}), where
	\begin{figure*}[t]
		\centering
		\begin{align}
		&\bar{R}_{\text{sum}}=\sum_{\scriptstyle k= 1}^K t^b_k\log_2 \left( t^b_k\sigma^2+\zeta s_kP_\text{max}\big( \text{Tr}(\boldsymbol{T}_k\tilde{\boldsymbol{\Phi}})+\kappa_1\big)\right)+t^o_k\log_2\left( t^o_k\sigma^2+z_k\left( \text{Tr}( \tilde{\boldsymbol{S}_k}\tilde{\boldsymbol{\Phi}})+|h_{\text{UM},k}|^2\right) \right)-\Xi_k,\label{wer}\\
		&(\text{\ref{p1-1}}): \big(\big(t^b_k-\tilde{ t}^b\big)a_k-\Pi_k\big)P_\text{max}\tilde{P}_\text{max}(\text{Tr}(\mathbf{U}_k\tilde{\boldsymbol{\Phi}}) +\kappa_2)+\big(a_kc_k(t^b_k\tilde{P}_\text{max}+\tilde{t}^bP_\text{max})\nonumber\\ \nonumber
		&\quad\quad-\Pi_k c_k(\tilde{P}_\text{max}-P_\text{max}) +b_kt^b_k-\tilde{t}^b\tilde{P}_\text{max}\big) ( \text{Tr}( \tilde{\boldsymbol{R}_k}\tilde{\boldsymbol{\Phi}})+|g_{\text{PU},k}|^2) +b_kc_k(t^b_k-\tilde{ P}^b)-c^2_k\Pi_k\geq 0,~\forall k.\label{ertt}
		\end{align}
		\medskip
		\hrule
	\end{figure*}
	$\Pi_k= P_{c,k}t^b_k+ \frac{z_k}{\delta}+ p_{c,k}t^o_k+\epsilon_kf_k^3T	+\frac{b_k(t^b_k-\tilde{t}^b)}{c_k}-Q_k$, $\tilde{t}^b=\sum_{\scriptstyle i= 1\atop\scriptstyle i \ne k}^K{t}^b_i$, and $\tilde{ P}_\text{max}=(1-\frac{s_k}{t^b_k})P_\text{max}$. Besides, $\tilde{\boldsymbol{S}_k}$ and $\tilde{\boldsymbol{R}_k}$ are two matrices with extra zero rows and columns. Finally, by dropping the rank-one constraint, i.e., $\text{Rank}(\tilde{\boldsymbol{\Phi}})=1$, the equivalent form of problem (P5) can be recast as 
	\begin{subequations}
		\begin{align}
		\text{(P6)}: ~~& \underset{ \tilde{\boldsymbol{\Phi}}}{\text{maximize}} \: \: \bar{R}_{\text{sum}}\\
		&\text{s.t.}\quad \text{(6b)},\label{32a}\\
		&\quad\quad t^b_k\log_2 \left( t^b_k\sigma^2+\zeta s_kP_\text{max}\big( \text{Tr}(\boldsymbol{T}_k\tilde{\boldsymbol{\Phi}})+\kappa_2\big)\right)\label{32b}\\ \nonumber
		& \quad\quad +t^o_k\log_2\left( t^o_k\sigma^2+z_k\left( \text{Tr}( \tilde{\boldsymbol{S}_k}\tilde{\boldsymbol{\Phi}})+|h_{\text{UM},k}|^2\right) \right)\\ \nonumber
		& \quad\quad -\Xi_k\geq \gamma_{\min,k },~\forall k,\\
		&\quad\quad \tilde{\boldsymbol{\Phi}}_{n,n}=1, \forall n \in \{1,...,N+2\}, \quad \tilde{\boldsymbol{\Phi}}\succeq \mathbf{0}.\label{32c}
		\end{align}
	\end{subequations}
	However, the obtained solution may not satisfy $\text{Rank}(\tilde{\boldsymbol{\Phi}})=1$, so we apply the following lemma to rewrite the non-convex rank-one constraint into its equivalent form\cite{Zargari1,Zargari11}.	
	\begin{lemma}
		The equivalent form of $\text{Rank}(\mathbf{V})= 1$, is given by
		\begin{equation}\label{rank}
		\|\mathbf{V}\|_*-\| \mathbf{V}\|_2 \leq 0.
		\end{equation}	
	\end{lemma}
	\begin{proof}
		For any $\mathbf{V}\in \mathbb{H}^m$, the inequality $\|\mathbf{V}\|_*=\sum_{j}\varepsilon_j\geq |\mathbf{V}|_2=\text{max}_j\{\varepsilon_j\}$ holds, where $\varepsilon_j$ is the $j$-th singular value of $\mathbf{V}$. Equality holds if and only if $\mathbf{V}$ has unit rank.
	\end{proof}
	However, constraint (\ref{rank}) is in the form of  DC functions. By adopting the first-order Taylor approximation of $\| \mathbf{V}\|_2$, we obtain
	\begin{align}\label{39}
	&\|\mathbf{V}\|_*-\|\mathbf{V}^{i}\|_2-\text{Tr}\left[\boldsymbol{\lambda}_\text{max}\left(\mathbf{V}^{i} \right) \boldsymbol{\lambda}^H_\text{max}\left(\mathbf{V}^{i} \right)(\mathbf{V}-\mathbf{V}^{i} )\right]\leq 0,
	\end{align}	
	where $\boldsymbol{\lambda}_\text{max}$ is the eigenvector corresponding to the maximum eigenvalue of matrix $\mathbf{V}^{i}$ in the $i$-th iteration. In this way, the equivalent form of (P6) can be stated as 
	\begin{subequations}
		\begin{align}
		\text{(P7)}: ~~& \underset{ \tilde{\boldsymbol{\Phi}}}{\text{maximize}} \: \: \bar{R}_{\text{sum}}-\Delta \tilde{f} \\
		&\text{s.t.}\quad \text{(\ref{32a})--(\ref{32c})},
		\end{align}
	\end{subequations}
	where 
	\begin{equation}
	  \tilde{f} =	 \|\tilde{\boldsymbol{\Phi}}\|_*-\|\tilde{\boldsymbol{\Phi}}^{i}\|_2-\text{Tr}\left[\boldsymbol{\lambda}_\text{max}\left(\tilde{\boldsymbol{\Phi}}^{i} \right) \boldsymbol{\lambda}^H_\text{max}\left(\tilde{\boldsymbol{\Phi}}^{i} \right)(\tilde{\boldsymbol{\Phi}}-\tilde{\boldsymbol{\Phi}}^{i} )\right]\!,  
	\end{equation}
	and $\Delta >0$ is a constant which penalizes the objective function for any $\boldsymbol{\Phi}$ whose rank is greater than one. Now,  (P7) is a convex problem that can be solved by standard convex optimization solvers such as CVX \cite{cvx}. The overall steps  for solving (P7) are  summarized in Algorithm 1.

		\begin{algorithm}[t]
		\algsetup{linenosize=\scriptsize }
		\small 
		\caption{Majorization minimization (MM) Algorithm}
		\begin{algorithmic}[1]
			\renewcommand{\algorithmicrequire}{\textbf{Input:}}
			\renewcommand{\algorithmicensure}{\textbf{Output:}}
			\REQUIRE Set the number of iterations ${t}$, the maximum number of iteration ${T}_\text{max}$, and $\boldsymbol{\theta}_0^{(0)}$.
			\STATE \textbf{repeat}\\
			\STATE \quad Obtain $\boldsymbol{T}_k$ and $\boldsymbol{U}_k$ according to (\ref{32}) and (\ref{332}), respectively.
			\STATE \quad Solve problem (P7) to
			obtain $\boldsymbol{\theta}$ by performing singular  \\ \quad value decomposition (SVD).\\
			\STATE \quad Set ${t}:={t}+1$;
			\STATE \textbf{until} ${t}={T}_{\text{max}}$
			\STATE \textbf{Return} $\{\boldsymbol{\theta}^*\}=\{\mathcal{\boldsymbol{\theta}}^{(t)}_0\}$
		\end{algorithmic}
	\end{algorithm}
	
	\subsection{Proposed Solution for Energy Minimization}
	Next, by setting $\alpha=0$, we solve the energy minimization problem to obtain the optimal value of $E_{\text{total,max}}$ which is formulated as follows:
	\begin{subequations}
		\begin{align}
		\text{(P8)}: ~~& \underset{\mathbf{t},\mathbf{s},\mathbf{z},\mathbf{f}}{\text{maximize}} \: \:-\bar{E}_{\text{total}}\\
		&\text{s.t.}\quad \text{(\ref{p3}), (\ref{p4-3})--(\ref{p4-6})}.
		\end{align}
	\end{subequations}
	Since problem  (P8) is also convex, we apply the Lagrange duality to obtain its optimal solutions in closed-form expressions next. 
	\begin{theorem}
		Given the non-negative Lagrange multipliers,
		i.e., $\{\lambda,\mu_k,\omega_k,\psi_k,\nu_k\}$, the maximizer of Lagrangian function is given by
		\begin{subequations}\label{lagrang2}
			{\small		\begin{align}
			&p^*_k=\left[\frac{\psi_k\delta}{(1+\nu_k)\ln2}-\frac{\sigma^2}{|h_k|^2}\right]^+,\; f^*_k=\left[\sqrt{\frac{\mu_kC_\text{cpu}-\psi_kT }{3\epsilon_kT(\nu_k+1)} }\right]^+,\\
			&\rho^*_k=
			\begin{cases}
			\left[\frac{Y_k-\sqrt{Y^2_k-4X_kZ_k}}{2X_k}\right]^+, & \omega_k=0,\\
			0, & \omega_k>0,
			\end{cases}
			\end{align}}
		\end{subequations}
		where $	X_k = \frac{\zeta \psi_k P^3_{\max}|g_k|^6|h_k|^2}{\ln2},$
		$Y_k = 2X_k+\frac{2\zeta c_k\psi_k P_{\max}|g_k|^4|h_k|^2}{\ln2}+\zeta (\nu_k)(a_kc_k-b_k)P_{\max}|g_k|^2|h_k|^2,$
		$Z_k = X_k+\frac{\zeta c^2_k\psi_k P_{\max}|g_k|^2|h_k|^2}{\ln2}+\frac{2\zeta c_k\psi_k P^2_{\max}|g_k|^4|h_k|^2}{\ln2}-(\nu_k)(a_kc_k-b_k)\sigma^2.$
	\end{theorem}
	\begin{proof}
	The Lagrangian function of problem (P8) is
given by (\ref{lagrang22}),
\begin{figure*}[h!] 
	\centering
	\begin{align}
	&\mathcal{L}(\lambda,\mu_k,\omega_k,\psi_k,\nu_k)=\lambda\bigg[ T-\sum_{k\in\mathcal{K}}({t^b_k+t^o_k})\bigg] +\sum_{k\in\mathcal{K}}\mu_k[f_\text{max}-f_k]+\sum_{k\in\mathcal{K}}\omega_k[t^b_k-s_k]\nonumber\\&+\sum_{k\in\mathcal{K}}\psi_k\bigg[ t^b_k\log_2(1+\frac{\zeta s_kP_\text{max}|h_k|^2|g_k|^2}{t^b_k\sigma^2})+t^o_k\log_2(1+\frac{z_k|h_k|^2}{t^o_k\sigma^2})+\frac{Tf_k}{C_\text{cpu}}-\gamma_{\min,k}\bigg] \nonumber\\& -\sum_{k\in\mathcal{K}}\nu_k\bigg[ \frac{z_k}{\delta}+ p_{c,k}t^o_k+\epsilon_kf_k^3T+P_{c,k}t^b_k-\left(\frac{a_k(1-\frac{s_k}{t^b_k})P_\text{max}g_k+b_k}{(1-\frac{s_k}{t^b_k})P_\text{max}g_k+c_k}-\frac{b_k}{c_k}\right)t^b_k-\sum_{i\in\mathcal{K},i\neq k}\bar{P}^b_k t^b_i-Q_k\bigg]\nonumber\\&- \sum\limits_{\scriptstyle k= 1}^K\left\lbrace P_{c,k}t^b_k+ \frac{z_k}{\delta}+ p_{c,k}t^o_k+\epsilon_kf_k^3T\right\rbrace.\label{lagrang22}
	\end{align}
	\medskip
	\hrule
\end{figure*}
where $\lambda$, $\mu_k$, $\omega_k$, $\psi_k$, and $\nu_k$ are the non-negative Lagrange multipliers associated with the constraints of problem (P8). Similar to (\ref{sss_k})--(\ref{sss_k2}), by taking the partial derivative of $\mathcal{L}(\lambda,\mu_k,\omega_k,\psi_k,\nu_k)$ with respect to the optimization variables in (P8), and setting $\frac{\partial\mathcal{L}}{\partial{f_k}}=0$ and $\frac{\partial\mathcal{L}}{{\partial{z_k}}}=0$, the optimal value of CPU frequency and transmit power of WD $k$ in the second phase can be obtained  as
\begin{align}
&\frac{\partial\mathcal{L}}{\partial{f_k}}=0\Rightarrow f^*_k=\left[\sqrt{\frac{\mu_kC_\text{cpu}-\psi_kT }{3\epsilon_kT(\nu_k+1)} }\right]^+,\nonumber\\& \frac{\partial\mathcal{L}}{{\partial{z_k}}}=0\Rightarrow z^*_k=\left[\frac{\psi_kt^o_k\delta}{(1+\nu_k)\ln2}-\frac{t^o_k\sigma^2}{|h_k|^2}\right]^+,~\forall k,
\end{align} 
respectively. Consequently, by using $z_k=p_kt^o_k$, the optimal transmit power can be written as $p^*_k=\left[\frac{\psi_k\delta}{(1+\nu_k)\ln2}-\frac{\sigma^2}{|h_k|^2}\right]^+,~\forall k.$
Base on Theorem 1, we have $X_k(\rho^*_k)^2- Y_k\rho^*_k+ Z_k=0$ for obtaining $\rho_k$, which yields
$X_k = \frac{\zeta \psi_k P^3_{\max}|g_k|^6|h_k|^2}{\ln2},$
$Y_k= 2X_k+\frac{2\zeta c_k\psi_k P_{\max}|g_k|^4|h_k|^2}{\ln2}+\zeta (\nu_k)(a_kc_k-b_k)P_{\max}|g_k|^2|h_k|^2$, and $
Z_k = X_k+\frac{\zeta c^2_k\psi_k P_{\max}|g_k|^2|h_k|^2}{\ln2}+\frac{2\zeta c_k\psi_k P^2_{\max}|g_k|^4|h_k|^2}{\ln2}-(\nu_k)(a_kc_k-b_k)\sigma^2,~\forall k.$
Since $0 \leq \rho^*_k\leq 1$ and 	$\frac{Y_k+\sqrt{Y^2_k-4X_kZ_k}}{2X_k}>1$, the optimal value of $\rho^*_k$ can be expressed as
\begin{align}
\rho^*_k=
\begin{cases}
\left[\frac{Y_k-\sqrt{Y^2_k-4X_kZ_k}}{2X_k}\right]^+, & \omega_k=0,\\
0, & \omega_k>0.
\end{cases}
\end{align} 
Next,  we obtain the optimal values of $t^b_k$ and $t^o_k$ efficiently, by applying the simplex method.
{By using the gradient method, the dual variables are updated as follows: }
{\begin{align}
       \lambda(t+1)&=\left[ \lambda(t)- \hat{\varkappa}(t) \left(  \sum_{k\in\mathcal{K}}(T-{t^b_k+t^o_k}) \right)
     \right]^+,~\forall t,\\
       \mu_k(t+1)&=\left[  \mu_k(t)- \hat{\varkappa}(t) \left(  f_\text{max}-f_k\right)
     \right]^+,~\forall t, k,\\
       \omega_k(t+1)&=\left[  \omega_k(t)- \hat{\varkappa}(t) \left(t^b_k-s_k\right)
     \right]^+,~\forall t, k,\\
     \nu_k (t+1)&=\bigg[   \nu_k(t)- \hat{\varkappa}(t)  \bigg(Q_k -\frac{z_k}{\delta}- p_{c,k}t^o_k-\epsilon_kf_k^3T\nonumber\\
     &+ \bigg(\frac{a_k(1-\frac{s_k}{t^b_k})P_\text{max}|g_k|^2+b_k}{(1-\frac{s_k}{t^b_k})P_\text{max}|g_k|^2+c_k}-\frac{b_k}{c_k} \bigg)t^b_k-P_{c,k}t^b_k \nonumber\\
     &+\sum_{i\in\mathcal{K},i\neq k}\bar{P}^b_k t^b_i \bigg)  \bigg]^+,~\forall t, k,\\
      \psi_k(t+1)&=\bigg[ \psi_k(t)- \hat{\varkappa}(t) \bigg( t^b_k\log_2(1+\frac{\zeta s_kP_\text{max}|h_k|^2|g_k|^2}{t^b_k\sigma^2})\nonumber\\
     &+t^o_k\log_2(1+\frac{z_k|h_k|^2}{t^o_k\sigma^2})+\frac{Tf_k}{C_\text{cpu}}-\gamma_{\min,k}\bigg)
    \bigg]^+,~\forall t, k,
\end{align}} {where $ \hat{\varkappa}(t) \geq 0$ is a positive step size at iteration $t$. Then, the dual variables  $\{ \lambda,\mu_k,\omega_k,\psi_k,\nu_k\}$  and the primal variables $\{\mathbf{t},\boldsymbol{\rho},\mathbf{p},\boldsymbol{\tau},\mathbf{f}\}$ are optimized iteratively.} 
	\end{proof}

	\subsection{Solution of Pareto Optimal System EE}
	After obtaining the optimal solutions of $R_{\text{sum,max}}$ and $E_{\text{total,max}}$, we proceed to obtain the Pareto optimal EE via solving problem (P3). By adopting the Lagrangian method, we give the optimal solutions in Theorem 3. 
	\begin{theorem}
		Given the non-negative Lagrange multipliers,
		i.e., $\{\lambda,\mu_k,\omega_k,\psi_k,\nu_k,\varsigma,\Omega\}$, the maximizer of Lagrangian function is given by
		\begin{subequations}\label{lagrang3}
		{\small	\begin{align}
			&p^*_k=\!\left[\frac{(\psi_k\!+\! \varsigma)\delta}{(\Omega\!+\!\nu_k)\ln2}\!-\!\frac{\sigma^2}{|h_k|^2}\right]^+\!\!\!,\; f^*_k=\left[\sqrt{\frac{\mu_kC_\text{cpu}-(\psi_k\!+\!\varsigma )T}{3\epsilon_kT(\nu_k\!+\!\Omega)} }\right]^+\!\!\!,\\
			&\rho^*_k=
			\begin{cases}
			\left[\frac{(\psi_k+ \varsigma)\delta}{(\Omega+\nu_k)\ln2}-\frac{\sigma^2}{|h_k|^2}\right]^+, & \omega_k=0\\
			0, & \omega_k>0,
			\end{cases}
			\end{align}}
		\end{subequations}	
		where $	X_k = \frac{\zeta (\psi_k+ \varsigma) P^3_{\max}|g_k|^6|h_k|^2}{\ln2} ,$
		$Y_k = 2X_k+\frac{2\zeta c_k(\psi_k+ \varsigma) P_{\max}|g_k|^4|h_k|^2}{\ln2}+\zeta (\nu_k+ \Omega)(a_kc_k-b_k)P_{\max}|g_k|^2|h_k|^2,$
		$Z_k = X_k+\frac{\zeta c^2_k(\psi_k+ \varsigma) P_{\max}|g_k|^2|h_k|^2}{\ln2}+\frac{2\zeta c_k(\psi_k+ \varsigma) P^2_{\max}|g_k|^4|h_k|^2}{\ln2}-(\nu_k+ \Omega)(a_kc_k-b_k)\sigma^2.$
	\end{theorem}
\begin{proof}
The Lagrangian function of problem (P3) is given by (\ref{lagrang33}),
\begin{figure*}[h!] 
	\centering
	\begin{align}
	&\mathcal{L}(\lambda,\mu_k,\omega_k,\psi_k,\nu_k,\varsigma,\Omega)=\chi+\lambda\bigg[ T-\sum_{k\in\mathcal{K}}({t^b_k+t^o_k})\bigg] +\sum_{k\in\mathcal{K}}\mu_k[f_\text{max}-f_k]+\sum_{k\in\mathcal{K}}\omega_k[t^b_k-s_k]\nonumber\\&+\sum_{k\in\mathcal{K}}\psi_k\bigg[ t^b_k\log_2(1+\frac{\zeta s_kP_\text{max}|h_k|^2|g_k|^2}{t^b_k\sigma^2})+t^o_k\log_2(1+\frac{z_k|h_k|^2}{t^o_k\sigma^2})+\frac{Tf_k}{C_\text{cpu}}-\gamma_{\min,k}\bigg] \nonumber\\& -\sum_{k\in\mathcal{K}}\nu_k\bigg[ \frac{z_k}{\delta}+ p_{c,k}t^o_k+\epsilon_kf_k^3T+P_{c,k}t^b_k-\left(\frac{a_k(1-\frac{s_k}{t^b_k})P_\text{max}g_k+b_k}{(1-\frac{s_k}{t^b_k})P_\text{max}g_k+c_k}-\frac{b_k}{c_k}\right)t^b_k\nonumber\\&-\sum_{i\in\mathcal{K},i\neq k}\bar{P}^b_k t^b_i-Q_k\bigg]+\varsigma\bigg[\sum_{k} \bigg\{t^b_k\log_2(1+\frac{\zeta s_kP_\text{max}|h_k|^2|g_k|^2}{t^b_k\sigma^2})+t^o_k\log_2(1+\frac{z_k|h_k|^2}{t^o_k\sigma^2})+\frac{Tf_k}{C_\text{cpu}}\bigg\}\nonumber\\&-R^m_{\text{sum,max}} +\chi\frac{|R^m_{\text{sum,max}}|}{\alpha_m}\bigg]-\Omega\bigg[\sum_{k\in\mathcal{K}} \bigg\{P_{c,k}t^b_k+\frac{z_k}{\delta}+ p_{c,k}t^o_k+\epsilon_kf_k^3T\bigg\}-E^m_{\text{total,max}} -\chi\frac{|E^m_{\text{total,max}}|}{\beta_m}\bigg].\label{lagrang33}
	\end{align}
	\medskip
	\hrule
\end{figure*}
where $\lambda$, $\mu_k$, $\omega_k$, $\psi_k$, $\nu_k$, $\varsigma$, and $\Omega$ are the non-negative Lagrange multipliers associated with the constraints of problem (P3). The optimal solutions can be obtained similar to Theorem 1 and 2 by taking the partial derivative of $\mathcal{L}(\lambda,\mu_k,\omega_k,\psi_k,\nu_k,\varsigma,\Omega)$ with respect to the optimization variables in (P4). By setting $\frac{\partial\mathcal{L}}{\partial{f_k}}=0$, $\frac{\partial\mathcal{L}}{{\partial{z_k}}}=0$, and using $z_k=p_kt^o_k$, the optimal value of CPU frequency and transmit power of WD $k$ can be expressed as
		\begin{align}
		&\frac{\partial\mathcal{L}}{\partial{f_k}}=0\Rightarrow f^*_k=\left[\sqrt{\frac{\mu_kC_\text{cpu}-(\psi_k+\varsigma )T}{3\epsilon_kT(\nu_k+\Omega)} }\right]^+,\\
		&\frac{\partial\mathcal{L}}{{\partial{z_k}}}=0\Rightarrow z^*_k=\left[\frac{(\psi_k+ \varsigma)t^o_k\delta}{(\Omega+\nu_k)\ln2}-\frac{t^o_k\sigma^2}{|h_k|^2}\right]^+,\\
		&p^*_k=\left[\frac{(\psi_k+ \varsigma)\delta}{(\Omega+\nu_k)\ln2}-\frac{\sigma^2}{|h_k|^2}\right]^+.
		\end{align} 
		Base on Theorem 1, by utilizing $X_k(\rho^*_k)^2- Y_k\rho^*_k+ Z_k=0$, we have $X_k  = \frac{\zeta (\psi_k+ \varsigma) P^3_{\max}|g_k|^6|h_k|^2}{\ln2} $,  $Y_k  = 2X_k+\frac{2\zeta c_k(\psi_k+ \varsigma) P_{\max}|g_k|^4|h_k|^2}{\ln2}+\zeta (\nu_k+ \Omega)(a_kc_k-b_k)P_{\max}|g_k|^2|h_k|^2$, and $Z_k = X_k+\frac{\zeta c^2_k(\psi_k+ \varsigma) P_{\max}|g_k|^2|h_k|^2}{\ln2}+\frac{2\zeta c_k(\psi_k+ \varsigma) P^2_{\max}|g_k|^4|h_k|^2}{\ln2}-(\nu_k+ \Omega)(a_kc_k-b_k)\sigma^2$,
		which yields the optimal value of $ \rho^*_k $ as follows
		\begin{align}
		\rho^*_k=\left[\frac{(\psi_k+ \varsigma)\delta}{(\Omega+\nu_k)\ln2}-\frac{\sigma^2}{|h_k|^2}\right]^+.
		\end{align} 
		Then,  the optimal values of $t^b_k$ and $t^o_k$ are obtained based on the simplex method. {By applying the gradient method, the updated dual variables are given by}
\begin{align}
       \lambda(t+1)&=\left[ \lambda(t)- \tilde{\varkappa}(t) \left(  \sum_{k\in\mathcal{K}}(T-{t^b_k+t^o_k}) \right)
     \right]^+,~\forall t,\\
       \mu_k(t+1)&=\left[  \mu_k(t)- \tilde{\varkappa}(t) \left(  f_\text{max}-f_k\right)
     \right]^+,~\forall t, k,\\
        \omega_k(t+1)&=\left[  \omega_k(t)- \tilde{\varkappa}(t) \left(t^b_k-s_k\right)
     \right]^+,~\forall t, k,\\
    \nu_k (t+1) &=\bigg[   \nu_k(t)- \tilde{\varkappa}(t)  \bigg(Q_k -\frac{z_k}{\delta}- p_{c,k}t^o_k-\epsilon_kf_k^3T\nonumber\\
     &+ \bigg(\frac{a_k(1-\frac{s_k}{t^b_k})P_\text{max}|g_k|^2+b_k}{(1-\frac{s_k}{t^b_k})P_\text{max}|g_k|^2+c_k}-\frac{b_k}{c_k} \bigg)t^b_k-P_{c,k}t^b_k \nonumber\\
     &+\sum_{i\in\mathcal{K},i\neq k}\bar{P}^b_k t^b_i \bigg)  \bigg]^+,~\forall t, k,\\
      \psi_k(t+1)& =\bigg[  \psi_k(t)- \tilde{\varkappa}(t) \bigg( t^b_k\log_2(1+\frac{\zeta s_kP_\text{max}|h_k|^2|g_k|^2}{t^b_k\sigma^2})\nonumber\\
     &+t^o_k\log_2(1+\frac{z_k|h_k|^2}{t^o_k\sigma^2})+\frac{Tf_k}{C_\text{cpu}}-\gamma_{\min,k}\bigg)
    \bigg]^+\!,~\forall t, k,\\
       \varsigma(t+1)& \!=\!\bigg[ \! \varsigma(t)\!-\! \tilde{\varkappa}(t) \bigg(\! \sum_{k} \bigg\{t^b_k\log_2(\!1\!+\!\frac{\zeta s_kP_\text{max}|h_k|^2|g_k|^2}{t^b_k\sigma^2}\!)\nonumber\\
     &+t^o_k\log_2(1+\frac{z_k|h_k|^2}{t^o_k\sigma^2})+\frac{Tf_k}{C_\text{cpu}}\bigg\} -R^m_{\text{sum,max}} \nonumber\\
     &+\chi\frac{|R^m_{\text{sum,max}}|}{\alpha_m}\!\bigg)
    \bigg]^+,~\forall t,  \\
      \Omega(t+1)& =\bigg[  \Omega(t)- \tilde{\varkappa}(t) \bigg( E^m_{\text{total,max}} +\chi\frac{|E^m_{\text{total,max}}|}{\beta_m}\nonumber\\&-\sum_{k\in\mathcal{K}} \bigg\{P_{c,k}t^b_k+\frac{z_k}{\delta}+ p_{c,k}t^o_k+\epsilon_kf_k^3T\bigg\}\bigg)
    \bigg]^+,~\forall t,  
\end{align} where $ \tilde{\varkappa}(t) \geq 0$ is a positive step size at iteration $t$. Then, the dual variables  $\{\lambda,\mu_k,\omega_k,\psi_k,\nu_k,\varsigma,\Omega\}$  and the primal variables $\{\mathbf{t},\boldsymbol{\rho},\mathbf{p},\boldsymbol{\tau},\mathbf{f}\}$ are optimized iteratively. 
	\end{proof}
	
	 Algorithm  2 gives the overall steps for solving (P1). {The first step maximizes the system throughput by jointly optimizing resource allocations and  RIS phase shifts. However, since the optimization problem is non-convex, Algorithm 1 is used to split it into two subproblems. The first subproblem gets the resource allocations, including the time/power allocations, backscattering coefficients, local computing frequencies, and execution times in closed-form solutions. Then, a suboptimal solution for the RIS phase shifts is achieved in the second subproblem.}

    {In the second step, since the energy consumption minimization problem is independent of RIS phase shifts, it only obtains resource allocations in closed-form solutions.}

   Finally, in the third step, Pareto optimal system EE is obtained via solving problem (P3) for different values of $\alpha$ and $\beta$  with a step size of $0.1$ such that $\alpha+\beta=1$, which is also a convex problem with closed-form solutions.

\begin{theorem}
    The main problem (P4) is non-increasing as the objective function value improves over each iteration in step 1 of Algorithm 2. Thus, the proposed AO algorithm is guaranteed to converge.
\end{theorem}

{	\begin{proof}	
The main problem (P4) is divided into two subproblems which optimize the resource allocations $\{\mathbf{t},\boldsymbol{\rho},\mathbf{p},\boldsymbol{\tau},\mathbf{f}\}$ and RIS phase shifts ($\boldsymbol{\theta}$), via solving problems (P4) and (P7), respectively. Let us define the objective value of the following problem: 
	\begin{subequations}
		\begin{align}
		\text{(Q1)}: ~~& \underset{\mathbf{t},\mathbf{s},\mathbf{z},\mathbf{f}, \boldsymbol{\theta}}{\text{maximize}} \: \:\bar{R}_{\text{sum}}\\
		&\text{s.t.}\quad \text{\text{(\ref{p1-66})}}, \text{(\ref{p3}), (\ref{p4-3})--(\ref{p4-6})}.
		\end{align}
	\end{subequations}
as
$f(\mathbf{t},\mathbf{s},\mathbf{z},\mathbf{f}, \boldsymbol{\theta})$. First, with fixed variable $\boldsymbol{\theta}$, problem (P4) is a convex problem and $(\mathbf{t}^{(i+1)},\mathbf{s}^{(i+1)},\mathbf{z}^{(i+1)},\mathbf{f}^{(i+1)})$ are the optimal solutions that maximize the value of the objective function. Accordingly, we have
\begin{equation}\label{condition1}
   f(\mathbf{t}^{(i\!+\!1)},\mathbf{s}^{(i\!+\!1)},\mathbf{z}^{(i\!+\!1)},\mathbf{f}^{(i\!+\!1)},\boldsymbol{\theta}^{(i)})\!\geq\! f(\mathbf{t}^{(i)},\mathbf{s}^{(i)},\mathbf{z}^{(i)},\mathbf{f}^{(i)},\boldsymbol{\theta}^{(i)}).
\end{equation}
Next, by maximizing $f$ via solving (P7), we obtain a sub-optimal solution
for the  RIS phase shifts as $\boldsymbol{\theta}^{(i)}$ with given optimization variables $(\mathbf{t}^{(i+1)},\mathbf{s}^{(i+1)},\mathbf{z}^{(i+1)},\mathbf{f}^{(i+1)})$. Thus, it guarantees that
\begin{align}\label{condition2}
   &f(\mathbf{t}^{(i+1)},\mathbf{s}^{(i+1)},\mathbf{z}^{(i+1)},\mathbf{f}^{(i+1)},\boldsymbol{\theta}^{(i+1)})\ge\nonumber\\& \qquad\qquad\qquad\qquad f(\mathbf{t}^{(i+1)},\mathbf{s}^{(i+1)},\mathbf{z}^{(i+1)},\mathbf{f}^{(i+1)},\boldsymbol{\theta}^{(i)}).
\end{align}
According to (\ref{condition1}) and (\ref{condition2}), we can conclude that
\begin{align}\label{condition3}
   &f(\mathbf{t}^{(i+1)},\mathbf{s}^{(i+1)},\mathbf{z}^{(i+1)},\mathbf{f}^{(i+1)},\boldsymbol{\theta}^{(i+1)})\ge\nonumber\\& \qquad\qquad\qquad\qquad f(\mathbf{t}^{(i+1)},\mathbf{s}^{(i+1)},\mathbf{z}^{(i+1)},\mathbf{f}^{(i+1)},\boldsymbol{\theta}^{(i+1)}), 
\end{align}
which indicates that the objective values of (Q1) are monotonically increasing after each iteration in step 1 of Algorithm 1. Meanwhile,  the objective values of (Q1) are non-negative. As a result, the proposed AO algorithm is guaranteed to converge. On the other hand, based on the fact that the initial point of each iteration is the solution of the previous one, the algorithm continues running to achieve a better solution in each iteration. Indeed, the objective function increases in each iteration or remains unchanged until the convergence is satisfied. Thus, the proof is completed.
	\end{proof}	}
	
	\begin{algorithm}[t]
		\algsetup{linenosize=\scriptsize }
		\small 
		\caption{Iterative Resource Allocation Algorithm}
		\begin{algorithmic}[1]
			\renewcommand{\algorithmicrequire}{\textbf{Input:}}
			\renewcommand{\algorithmicensure}{\textbf{Output:}}
			\REQUIRE Set initial iterations $i=0$ and $j=0$, the maximum number of iterations $I_\text{max}$ and $J_\text{max}$, and Lagrangian variables vectors $\{\lambda,\mu_k,\omega_k,\Omega_k,\nu_k,\psi_k,\varsigma,\Omega\}$.\\ 
			\STATE \textbf{Step 1: Throughput and phase shifts optimization based on the AO algorithm}\\ 
			\STATE \textbf{Repeat}\\ 
			\STATE \quad\textbf{Repeat}\\
			\STATE \quad\quad Find optimal resource allocations using (\ref{lagrang1}).
			\STATE \quad\quad Update Lagrangian variables vectors,\\ \quad \quad $\{\lambda,\mu_k,\omega_k,\Omega_k,\nu_k\}$.
			\STATE \quad\quad Set $i:=i+1$;
			\STATE \quad \textbf{Until} $i=I_\text{max}$
			\STATE \quad For the given optimal resource allocations, solve problem\\ \quad  (P7) according to \textbf{Algorithm 1}.\\
			\STATE \quad Set $j:=j+1$;
			\STATE \textbf{Until} $j=J_\text{max}$; Obtain $\{\mathbf{t}^*,\boldsymbol{\rho}^*,\mathbf{p}^*,\boldsymbol{\tau}^*,\mathbf{f}^*,\boldsymbol{\theta}^*,R^*_{\text{sum,max}}\}$.
			\STATE \textbf{Step 2: Energy optimization}\\ 
			\STATE \textbf{Repeat}\\
			\STATE \quad Find optimal resource allocations using (\ref{lagrang2}).
			\STATE \quad Update Lagrangian variables vectors, $\{\lambda,\mu_k,\omega_k,\psi_k,\nu_k\}$.
			\STATE \quad Set $i:=i+1$;
			\STATE \textbf{Until} $i=I_\text{max}$; Obtain $E^*_{\text{total,max}}$.
			\STATE \textbf{Step 3: Pareto optimal EE}\\ 
			\STATE For given $\{R^*_{\text{sum,max}}, E^*_{\text{total,max}}, \boldsymbol{\theta}^*\}$, solve problem (P3) for different values of $\alpha$ and $\beta$  with a step size of $0.1$ such that $\alpha+\beta=1$.
		\end{algorithmic}
	\end{algorithm}
	
	\section{Computational Complexity Analysis}

	In this section, we investigate the computational complexity of our proposed algorithm. Firstly, the order of complexity to obtain the phase shifts via the MM approach is $\mathcal{O}(T_{\max}(N+2)^6)$, where $T_{\max}$ is the number of iterations. Secondly, for the MOOP, the order of complexity is $\mathcal{O}(K^5)$. Let the subgradient algorithm take $\Lambda$ iterations to converge. Therefore, the order of complexity of MOOP is $\mathcal{O}(K^5 \Lambda)$ asymptotically. However, if CVX is used to solve the SOOPs, i.e., (P4) and (P8) and the MOOP, i.e., (P3), it applies the interior point method. Therefore, the computational complexity of the interior point method to obtain the optimal solutions for each problem in (P3),  (P4), and (P8) are given by $\mathcal{O}(\sqrt{m_1}\log(m_1))$, $\mathcal{O}(\sqrt{m_2}\log(m_2))$, and $\mathcal{O}(\sqrt{m_3}\log(m_3))$ \cite{Lu,Boyd}, respectively, where $m_1$, $m_2$, $m_3$ denote the number of inequality constraints of (P3), (P4), and (P8), respectively. On the other hand, the computation complexity of throughput and phase shifts optimization based the AO algorithm is given by $\mathcal{O}\left(J_{\max}\left(T_{\max}(N+2)^6+\sqrt{m_2}\log(m_2)\right)\right)$, where $J_{\max}$ is the number of iterations. Consequently, the total computational complexity for each step size can be expressed as  
	$\mathcal{O}(J_{\max}(T_{\max}(N+2)^6+\sqrt{m_2}\log(m_2)) + \sqrt{m_3}\log(m_3)+ \sqrt{m_1}\log(m_1)  )$.


	\section{Simulation Results}
    	This section presents numerical results to assess the performance of the proposed scheme. The simulation parameters are presented in Table I unless specified otherwise. According to the standard  path loss model, the channel gain is given by  $h=\bar{h}d^{-\upsilon},$ where $\bar{h}$ and $d$ are corresponding small-scale fading coefficients  and distance, respectively, and  $\upsilon$ denotes the path loss exponent. In particular, the path loss exponents of the PB-to-RIS, RIS-to-WD, and RIS-to-MEC channels are assumed to be dominated by the line of sight (LoS) link with the path loss exponent of $2.2$ \cite{Qingqing,Wu}, which are lower than that of the PB-to-WD and WD-to-MEC channels assumed equal to $3$. For the small-scale fading, we assume the  Rayleigh fading model.  However, the Ricain fading model is investigated later in Fig. \ref{N}, and the PFs of system throughput and energy consumption with random phase shift at the RIS are also studied later in Fig. \ref{fig3}. The distances between nodes are given in Fig. \ref{sim}. {To demonstrate the advantage of the proposed scheme, by considering baseline 1 as only BC mode (each WD fully offloads task bits only using BC) and baseline 2 as BC and local computing mode (each WD offloads task bits by only utilizing BC and also computing them locally), we compare Algorithm 2 against the following benchmarks: 
        1) Algorithm 2 with baseline 1;
        2) Algorithm 2 with baseline 2;
        3) No-RIS;
        4) No-RIS with baseline 1;
        5) No-RIS with baseline 2 .}

		\begin{table}
		\renewcommand{\arraystretch}{1.05}
		\centering
		\caption{\small Simulation Parameters}
		\label{table-notations}\scalebox{0.85}{
			\begin{tabular}{| l| l|  }
				\hline
				\textbf{Parameters}& \textbf{Values} \\\hline
				Number of reflecting elements, $N$ & $20$ \\ \hline  Minimum required computation bits,~$L_\text{min}$ &$20$ kbits\\ \hline
				Maximum CPU frequency, $f_{\text{max}}$ & $5 \times 10^5$ kHz \\ \hline   Maximum transmit power at the PB, $P_\text{max}$ &$1$ W\\ \hline  
				Effective capacitance coefficient, $\epsilon_k$ & $10^{-26}$\\ \hline   Number of CPU cycles & $1000$ \\ \hline
				Direct link path loss exponent& $3$ \\ \hline  Reflected link path loss exponent & $2.2$ \\ \hline	Offloading circuit power consumption,
				& $5$ mW \\ \hline  BC circuit power consumption, $P_{c,k}$& $0.1$ mW	\\ \hline Power amplifier efficiency, $\delta$ & $1$ \\ \hline  Primary energy for all WDs,~$Q$& $[1,~1,~0,~0]$ \\ \hline
				Maximum curvature, $l$& $2.5\times 10^{-16}$\\ \hline  Non-linear EH parameters&		\begin{tabular}[c]{@{}l@{}}$a_k=2.463,~b_k= 1.635,$ \\ \qquad\quad	$c_k=0.826$\end{tabular} \\ \hline
				Number of WDs& $4$ \\ \hline  Penalty factor,~$\Delta$ & $5\times 10^{5}$ \\ \hline
				Entire time block, $T$ & $1$ s \\ \hline  	Noise power,~$\sigma$& $-120$ dBm \\ \hline
				System bandwidth & $100$ kHz \\ \hline   Performance gap,~$\zeta$& $0.0316$\\ \hline
		\end{tabular}}
	\end{table}
	\begin{figure}[t]
		\centering
		\includegraphics[width=3.5in]{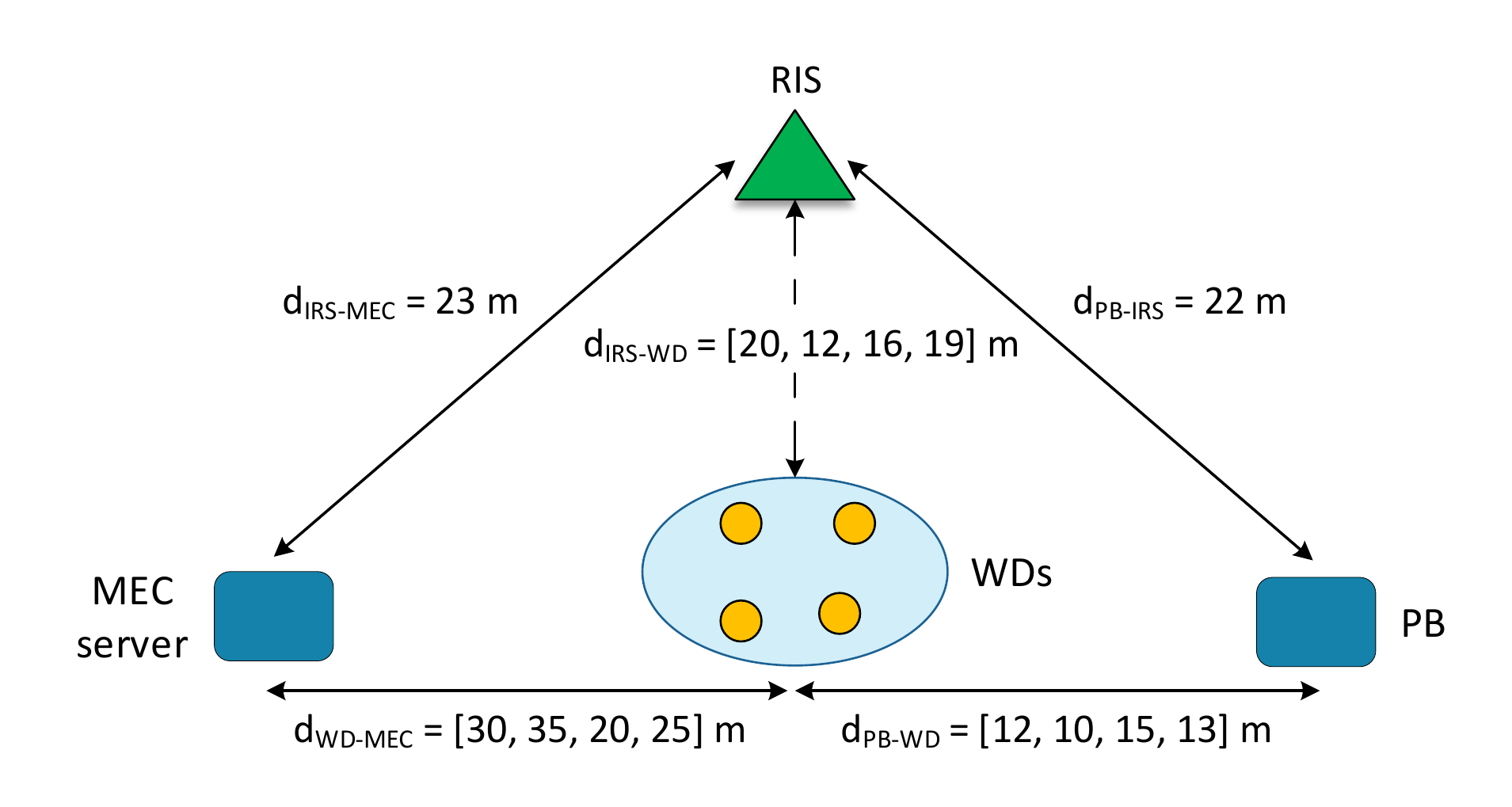}
		\caption{ {Simulation setup for the RIS-BC-MEC.}}\label{sim}	\vspace{-5mm}
	\end{figure} 
	
		\begin{figure}
		
		\begin{minipage}{.45\textwidth}
			\centering
			\includegraphics[width=3.5in]{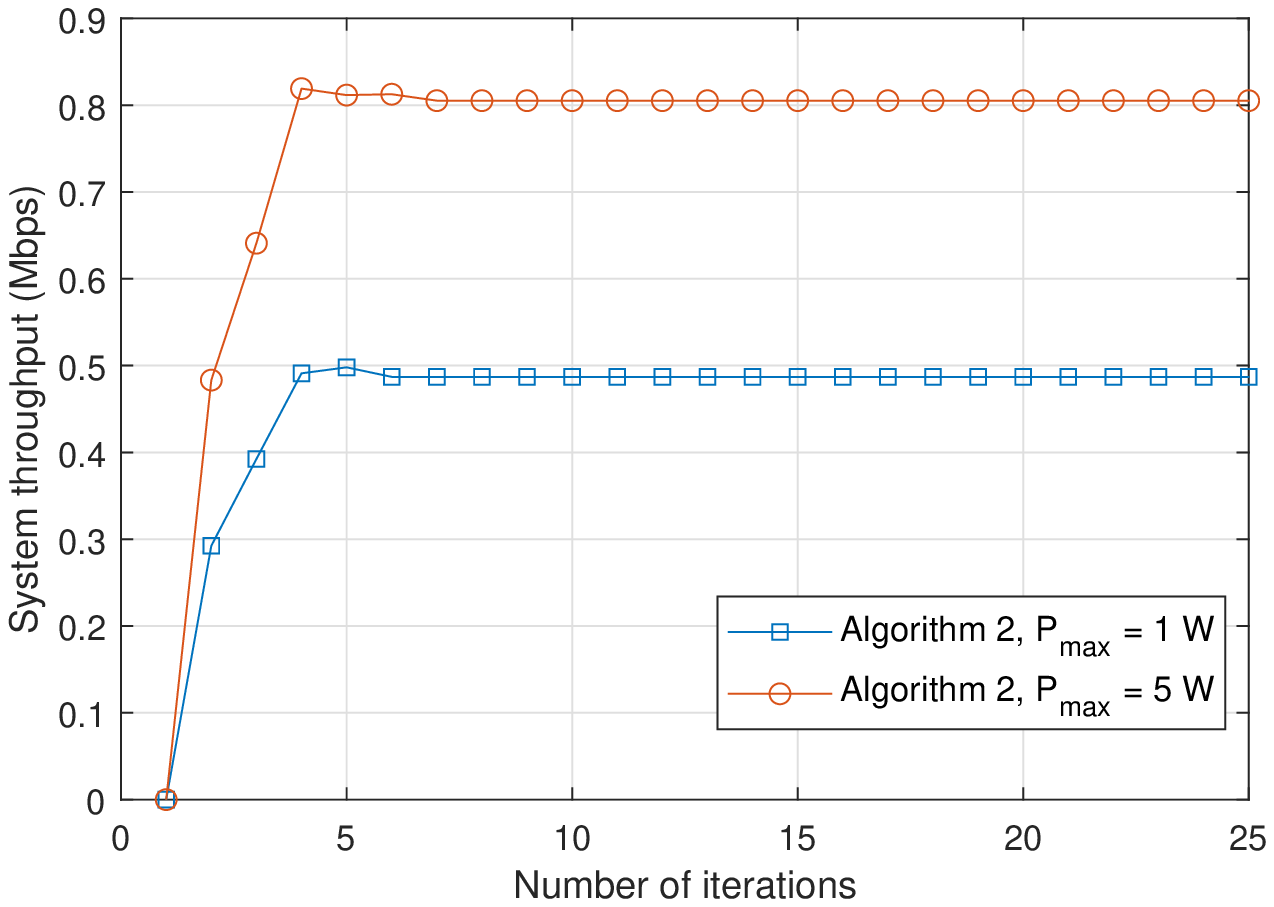}
			\caption{ Convergence behavior of Algorithm 2.}\label{fig66}
		\end{minipage}\hfill
		\begin{minipage}{.45\textwidth}
			\centering
			\includegraphics[width=3.5in]{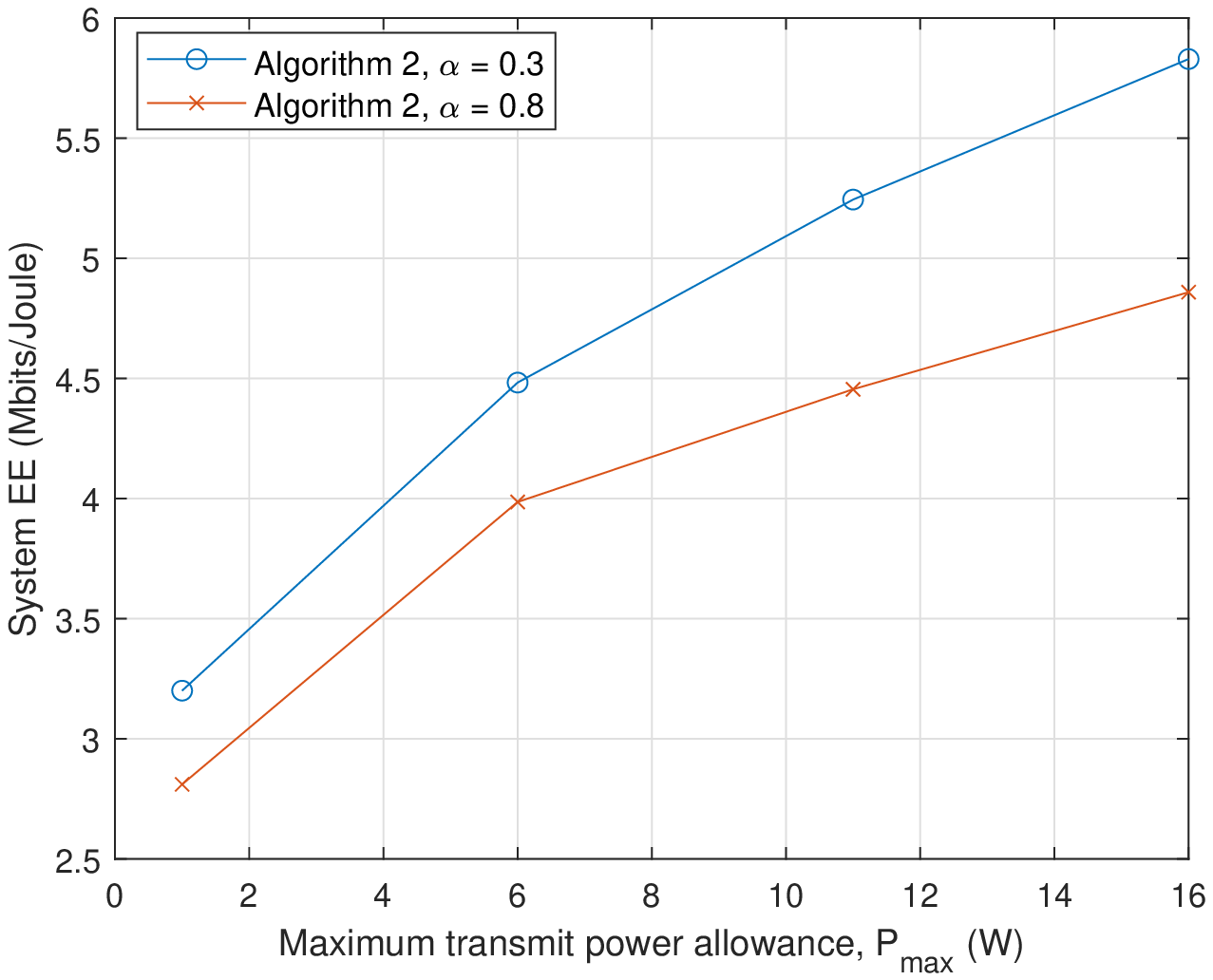}
			\caption{ System EE versus maximum transmit power.}\label{fig55}
		\end{minipage}
		\vspace{-6mm}
	\end{figure}
		\begin{figure}[!tbp]
		\centering
		\subfloat[ System EE versus weighting coefficient $\alpha$.]{\includegraphics[width=3.5in]{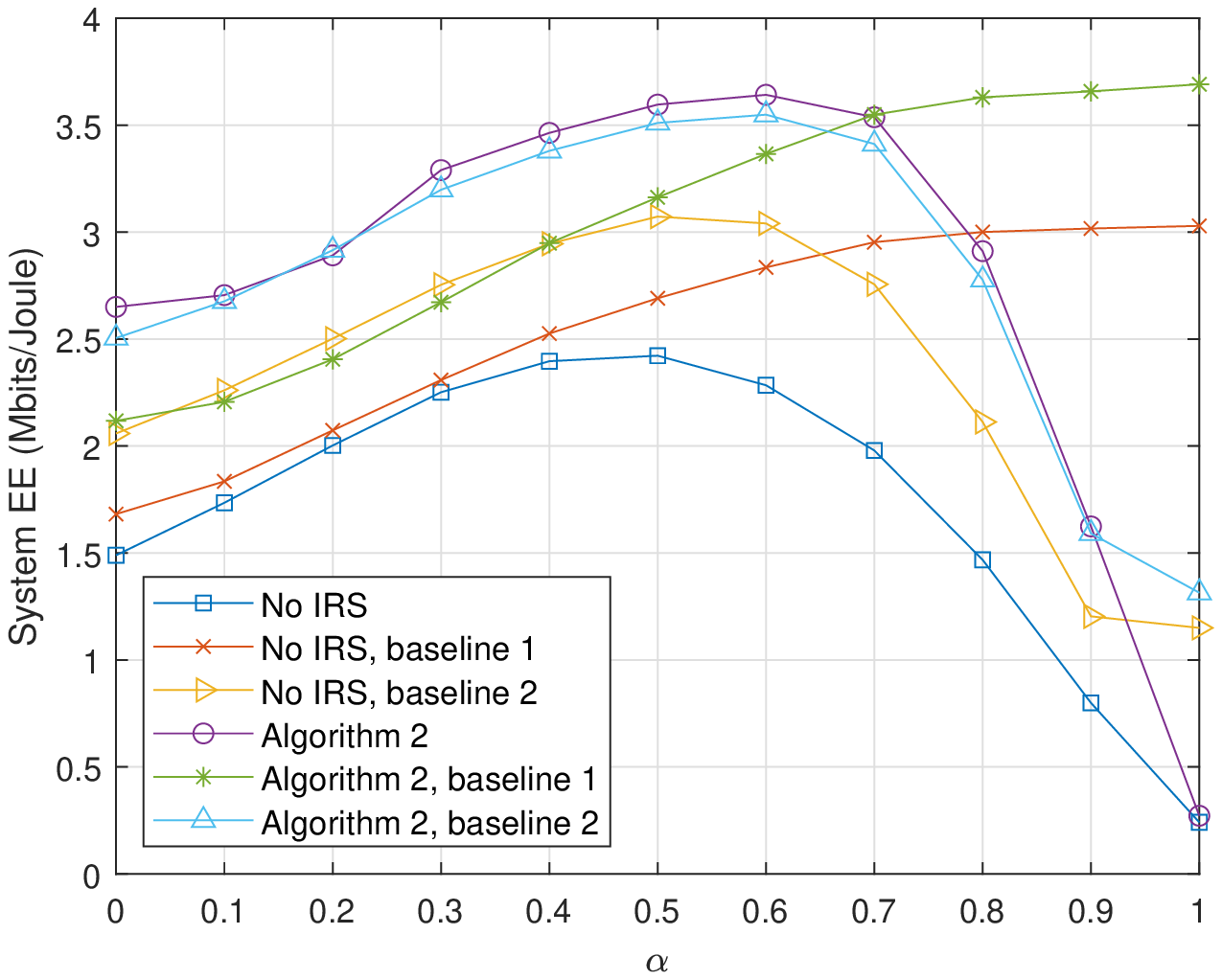}\label{fig1}}
		\hfill
		\subfloat[Throughput and energy consumption versus $\alpha$.]{\includegraphics[width=3.5in]{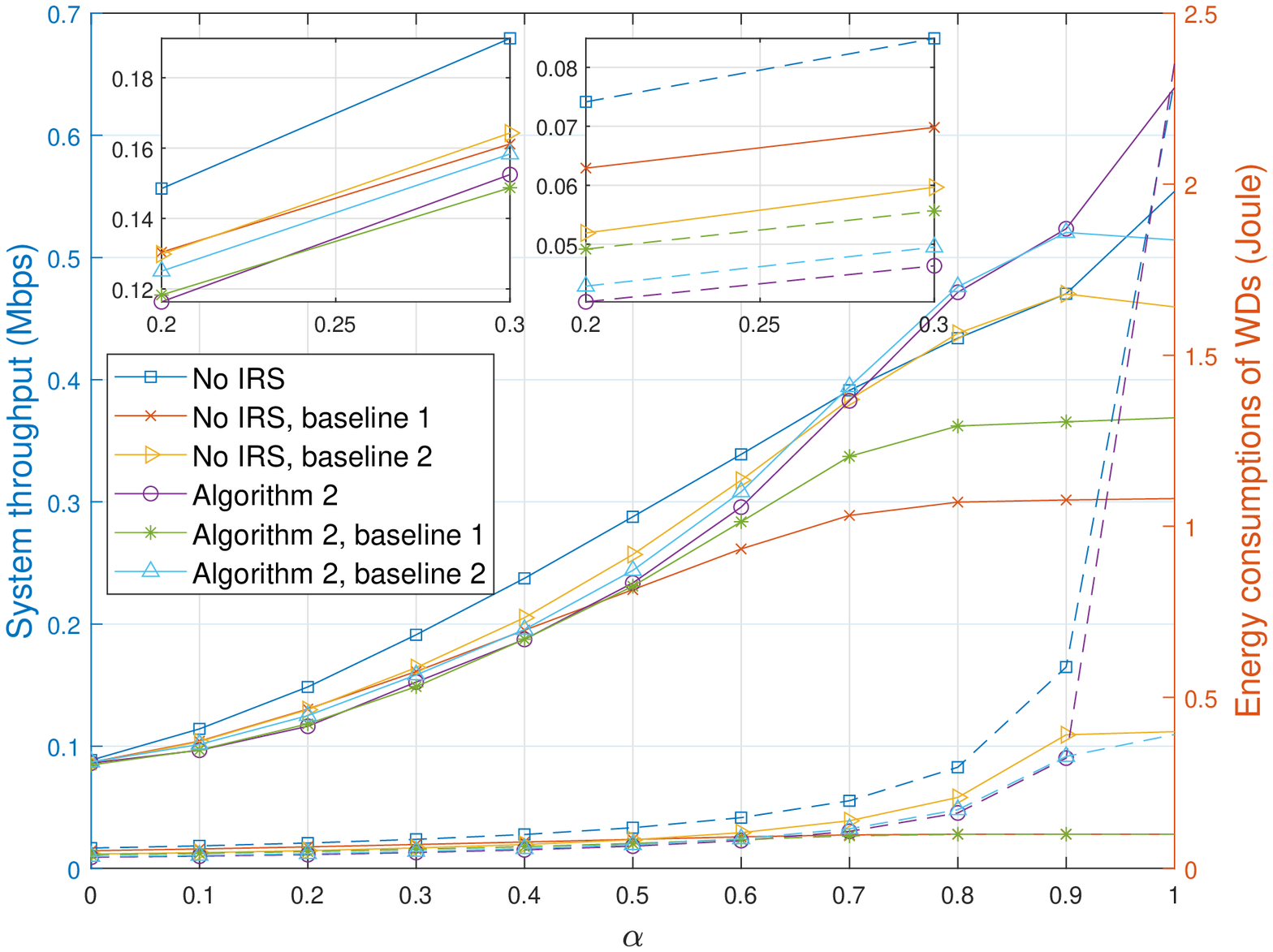}\label{fig5}}
		\caption{\small {The effect of weighting coefficient, $\alpha$, on system performance where the dash lines represent the energy consumption of WDs corresponding to the system throughput.}}\label{fig6}
		\vspace{-5mm}
	\end{figure}
	
	\subsection{Convergence Behavior}
	Fig. \ref{fig66} illustrates the convergence behavior of Algorithm 2 for different PB  transmit power values. All the curves converge to a stationary point within less than eight iterations on average. More than  $90\%$ of the system EE is reached within five iterations. This figure also reveals that the sum data rate is monotonically increasing at each iteration.
				\begin{figure}[tbp]
		\centering
		\subfloat[ System EE versus system throughput.]{\includegraphics[width=3.5in]{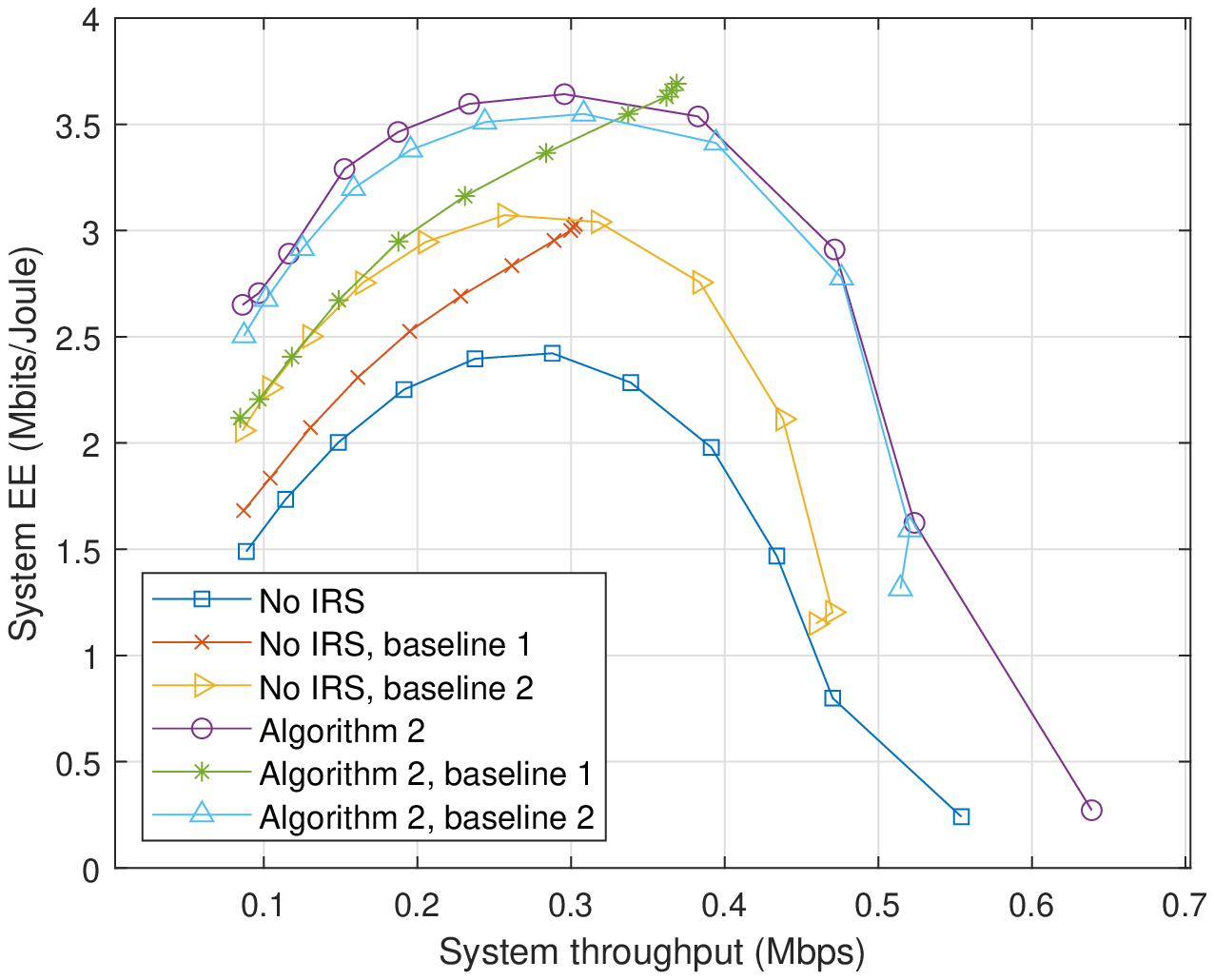}\label{fig2}}
		\hfill
		\subfloat[PFs of System throughput  and  energy consumption.]{\includegraphics[width=3.5in]{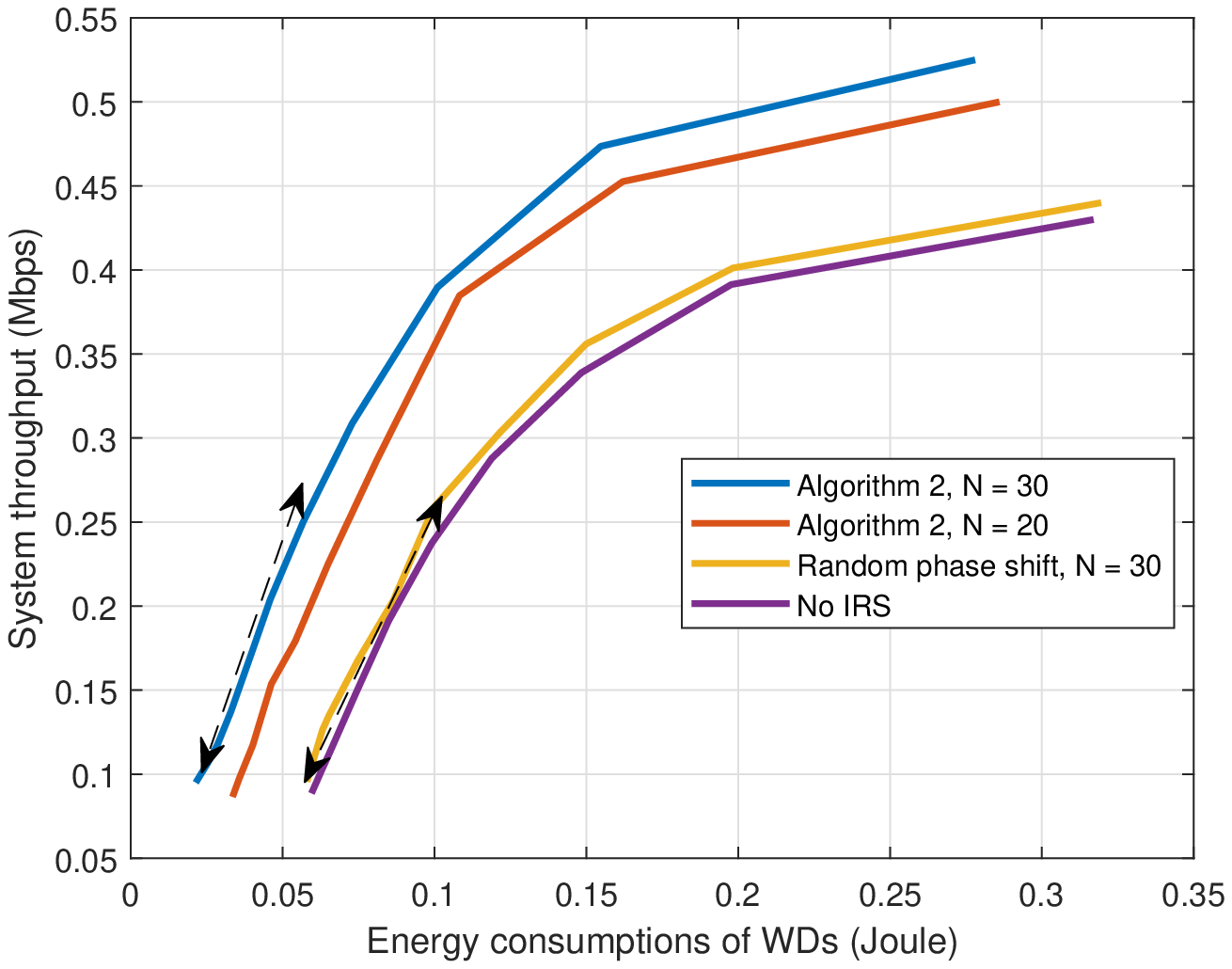}\label{fig3}}
		\caption{\small Trade-off regions.}\label{fig4}
			\vspace{-5mm}
	\end{figure}
	\subsection{System EE vs. Maximum Transmit Power}
	Figure \ref{fig55} shows the achieved system EE of Algorithm 2. As seen, increasing maximum transmit power at the PB improves the system EE. As revealed in Proposition 2, the system achieves the maximum EE when the PB transmits with maximum power. Thus, in the first phase, the achievable offloading throughput improves by increasing PB's transmit power as the throughput  is an increasing function of $P_{\max}$, equation (\ref{1}). 
	Additionally, each WD can harvest more energy since it enjoys more received power, which allows it to transmit information to the MEC in the second phase with increased transmission power.

In conclusion, the system throughput increases significantly at the high transmit power region than the consumed energy, leading to a higher system EE. To better illustrate the uptrend of the system EE, we have also considered two cases, i.e., $\alpha = 0.3$ and $\alpha = 0.8$. When the value of $\alpha$ is large, a stringent data rate is imposed on the system, leading to a higher aggregate energy consumption and, accordingly, lower system EE than the smaller $\alpha$ case.

	\subsection{System EE vs. Weighting Coefficient $\alpha$}
	Fig. \ref{fig1} shows system EE versus the weighting coefficient of the system throughput, $\alpha$, for different baseline schemes.  This figure shows that with increasing $\alpha$, the system EE first increases and then decreases for all the schemes. This anomaly arises because the increase of  $\alpha $ imposes a stringent data rate, increasing aggregate energy consumption, and thus reducing system EE. In contrast, for a smaller  $\alpha$, a substantial increase in system throughput requires slight energy consumption. Therefore,  the choice of $ \alpha $ has a critical effect on the system EE.   The maximum value of system EE changes from one case to another. More specifically, for the proposed scheme, the system EE is higher than the baseline schemes up to $0.7$, primarily when no-RIS is employed. This case arises because the RIS  improves WPT and data transfers in the first and second phases. In the first phase, the RIS assists the PB signal for WPT  and WDs to backscatter their bits.   In the second phase, it helps each WD in bit offloading to the MEC server in the AT mode.  We can also see that baseline two performs better than baseline one up to $0.6$ since the former includes local computing, which leads to higher system EE. However, we observe a slight increase in system EE for baseline one since baseline one considers only static energy consumption corresponding to BC. 	{For a low value of $ \alpha$, more  transmission time is allocated to a WD with a better channel condition,  enhancing the system EE. However,  the high $\alpha$ region imposes  a stringent data rate on the system. Therefore, regardless of the change in energy consumption, the rate of change in system throughput for baseline scheme 1 becomes relatively subtle, achieving the maximum system EE. This is caused by resource allocation budget limitations.}  In contrast, for baseline two, the energy consumption of the calculation bits, i.e., local computing, is considered besides the energy consumption of BC. 
	
	Fig. \ref{fig5} shows the trade-off between system throughput and energy consumption by changing the value of $\alpha$. This figure illustrates Fig. \ref{fig1} in more detail. As observed, the system throughput improves when the value of $\alpha$ increases. However, this results in more energy consumption, which thus decreases the system EE. Interestingly, the RIS leads to much less energy consumption in Algorithm 2, resulting in high system EE.

	\subsection{Trade-Off Regions}	
	{ Fig. \ref{fig4} reveals interesting results for trade-off regions with the  parametric plots, i.e., $[ \eta_{EE}(\alpha), R_{sum}(\alpha) ]$, where $\alpha$ varies from  0 to 1. } Fig. \ref{fig2} depicts the system EE versus the system throughput for the different baseline schemes. As the system throughput grows, the EE increases and then decreases. As the system EE is the ratio of the throughput to the total energy consumption and the system throughput is itself a function of transmit power, increasing the throughput leads to higher energy consumption. {As the system sum rate increases, WDs must compete for time resources more fiercely and thus increase their transmit powers during the  active transmission to meet their throughput requirements, causing a faster decay in the system EE.} Fig. \ref{fig2}  also shows that the proposed schemes with the RIS perform better than the no-RIS case.
	Another interesting phenomenon is that baseline one outperforms baseline two. This trend occurs because baseline one considers only the energy consumption of the BC mode, which shows that the harvested energy compensates for the energy consumption of this mode and leads to a better performance gain than baseline two.
	
	Fig. \ref{fig3} highlights the PFs of the system throughput and energy consumption for the cases with and without RIS. The PF includes all Pareto optimal solutions. Recall that this means there is no other policy that can improve one objective without detriment to the other objective. Along with the two PFs, any slight improvement in energy consumption for all cases increases the system throughput, which is desirable. The derivative of the PF  shows the marginal benefit of increased system throughput per unit increase in energy use.  Notice that this gradient increases with the number of RIS elements in comparison to the no-RIS case. In addition, we also study the effect of random phase shifts at the RIS by setting the components in $\boldsymbol{\theta}$ randomly in  $[0,2\pi]$ and then optimize the resource allocations. By applying the random phase shift, the obtained PF is closer to the no-RIS case since without optimizing phase shifts, the average signal power of the reflected signal is comparable to that of the signal from the direct link. Significantly,  the  RIS deployment improves the PF by reducing the total energy consumption and thus increasing the system throughput compared to the no-RIS and random-phase-shift cases.
	
		\begin{figure}
			\centering
			\includegraphics[width=3.5in]{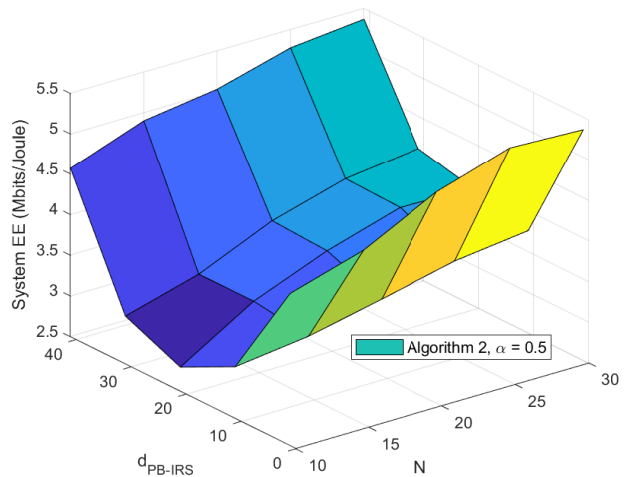}
			\caption{\small System EE versus $N$ and $d_{\text{PB-RIS}}.$ }\label{surf}
			\vspace{-6mm}
	\end{figure}
	
		\begin{figure}
			\centering
			\includegraphics[width=3.5in]{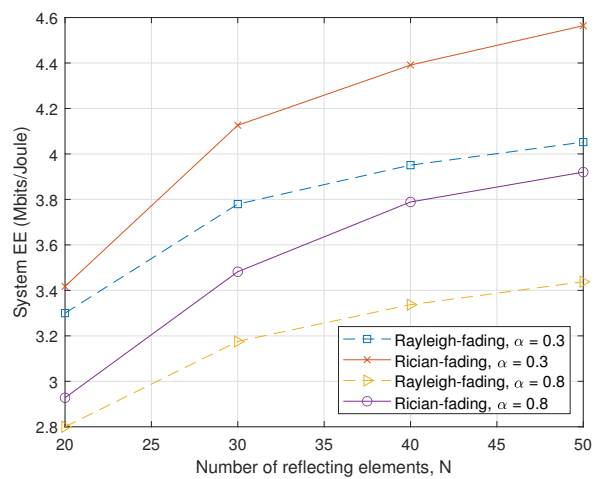}
			\caption{\small System EE versus number of reflecting elements, $N$.}\label{N}
			\vspace{-5mm}
	\end{figure}

	\subsection{Impact of Reflecting Elements and Distance}
	In Fig. \ref{surf}, the impact of reflecting elements, $N$, and distance, $d$, on the system EE is investigated for Algorithm 2  with $\alpha=0.5$. This surface is interesting since it provides the resource allocator with the downward and ascending trend of the system EE in terms of distance and number of reflecting elements. As can be seen, increasing the number of reflecting elements can increase system performance regardless of RIS distance with the PB and the MEC due to the contribution of the reflected link. On the other hand, the effect of distance between the RIS and the PB is U-shaped, indicating that the system performance increases when the RIS is close to either the MEC or the PB. However, if RIS is placed in the middle between the PB and the MEC, significant performance degradation is observed. This can be interpreted as follows: consider $d_{\text{PI}}$ and $d_{\text{IU},k}$ as the distances between the PB-RIS and RIS-WD $k$, respectively. When these two distances are equal, i.e., $d_{\text{PI}}=d_{\text{IU},k}$, the overall path loss that includes the PB and the MEC, $d_{\text{PI}}d_{\text{IU},k}$ is maximized, leading  to the largest path loss. As can be observed for $ N=10 $, moving the RIS away from $30$ m to $40$ m of the PB exhibits approximately $3$ (Mbits/Joule) to $4.5$ (Mbits/Joule) improvement in the system EE. This is because the RIS moves closer to the MEC,  resulting in a  lower path loss between the RIS and MEC server. As a summary:
	\begin{enumerate}
	    \item {Fig. 7a is similar to 6a. In Fig. 6a, we plot the $\eta_{EE}=\frac{R_{\text{sum}}}{E_{\text{total}}}$ versus $\alpha$. In Fig. 7a, we plot $\eta_{EE}(\alpha)=\frac{R_{\text{sum}}(\alpha)}{E_{\text{total}}(\alpha)}$ versus $R_{\text{sum}}(\alpha)$ for different values of $0\leq \alpha \leq 1$. As the value of $\alpha$ increases for green and red curves in Fig. 6a,  the EE saturates and reaches its maximum point. This is because increasing  $\alpha$  imposes a stringent data rate on the system. Further increasing  $\alpha$ then increases the EE marginally because of the resource limitations, i.e., time allocation and maximum transmit power of the PB. 
	    Thus, the red and green curves reach their maximum value for $0.7 \leq \alpha \leq 1$.}

	    \item {In Fig. 7a, for a high value of $\alpha$, similar to Fig 6a, the system EE and system throughput reach their maximum value. That is why the red and green curves end before $0.4$ Mbps. As observed, the maximum value of system EE for the red and green curves are $3$ (Mbps/joule) and $3.7$ (Mbps/joule), respectively, which is the same as Fig. 6a. Further, the maximum value of the system throughout for the red and green curves are $0.3$ (Mbps) and $0.38$ (Mbps), respectively.}
 
	\end{enumerate}

	\subsection{The impact of the number of  reflecting elements,  $N$}	
	Finally, Fig. \ref{N} reveals system EE versus the number of reflecting elements at the RIS for different values of $\alpha$. For comparison, we consider the Ricain fading model with LoS components for the channel links between the PB-RIS and the RIS-MEC  with a 10 dB  Rician factor  \cite{Zargariii}. We observe that the system EE for all schemes improves monotonically by increasing the number of reflecting elements at RIS, $N$.  Large $N$  provides a higher achievable data rate and much lower aggregated energy consumption, leading to a higher system EE. For significant values of $N$, the proposed scheme with the Ricain fading model has a high impact on the performance gain compared to the Rayleigh fading model since more reflecting elements with the LoS links contribute to a higher achievable data rate. These observations show that the RIS with no RF  chains provides a solid reflective channel link using lower-cost passive elements.
	
	\section{Conclusion}	
	
	This paper studied the joint time/power allocations, local computing frequency, execution time, and RIS phase shifts for a multi-WD WPT-based BC-MEC network supported by RIS, where a non-linear EH model was considered at each WD. In particular, the system EE was maximized subject to the following constraints: minimum required computation task bits, EH, total transmission time, maximum CPU frequency, backscattering coefficients, total transmit power, and phase shifts at the RIS. We considered the EE maximization problem as a MOOP and exploited the Tchebycheff method to transform it into two  SOOPs. To solve the SOOPs, we derived optimal closed-form solutions for the resource allocations,  applied SDR/MM techniques,  and introduced a penalty function to optimize the reflecting elements at the RIS. Finally, simulation results demonstrated that the RIS offers reliable performance even for a few reflecting elements and confirmed the effectiveness of the proposed algorithm. Furthermore, since this is the first manuscript to design resource allocations for WPT-based RIS-BC-MEC networks, it opens up many future research directions. The first extension would be to investigate  WDs that are equipped with multiple antennas. That opens up the possibility of beamforming and other multiple antenna techniques to enhance the system performance.  The second extension could be to maximize the EE under the imperfect CSI scenarios and to develop robust algorithms against CSI imperfections. There are many more future directions, and this list is not exhaustive.

\end{document}